\documentclass[11pt]{article}
\usepackage[margin=.8in,left=.8in]{geometry}
\usepackage{amsthm}
\usepackage{amssymb}
\usepackage{amsmath}
\usepackage{mathtools}
\usepackage{adjustbox}

\usepackage[table]{xcolor}

\usepackage{stmaryrd}    
\usepackage{xfrac}
\usepackage{enumitem}\setlist{nosep}

\usepackage[safe]{tipa} 

\usepackage{hhline}

\usepackage{tensor}

\usepackage{tikz}
\usetikzlibrary{
  backgrounds, 
  decorations.pathmorphing, 
  decorations.markings,
  decorations.text
}
\usetikzlibrary{cd}

\usepackage{mathrsfs} 
\DeclareMathAlphabet{\mathpzc}{OT1}{pzc}{m}{it} 

\usepackage{hyperref}
\hypersetup{
  colorlinks = true,
  linkcolor= darkblue, citecolor=darkblue            
}

\hypersetup{
        colorlinks=true,
        linkcolor=darkblue,
        filecolor=darkblue,      
        urlcolor=black,
        citecolor=darkblue,
        linktoc=page
    }

\usepackage{cleveref}

\DeclareRobustCommand{\rchi}{{\mathpalette\irchi\relax}}
\newcommand{\irchi}[2]{\raisebox{\depth}{$#1\chi$}} 

\definecolor{darkblue}{rgb}{0.05,0.25,0.65}
\definecolor{darkgreen}{RGB}{20,140,10}
\definecolor{lightgray}{rgb}{0.9,0.9,0.9}
\definecolor{darkorange}{RGB}{200,100,5}
\definecolor{darkyellow}{rgb}{.91,.91,0}
\definecolor{orangeii}{RGB}{200,100,5}
\definecolor{lightblue}{RGB}{243, 250, 255}

\newtheorem{theorem}{Theorem}[section]
\newtheorem{claim}{Claim}[section]
\newtheorem{lemma}[theorem]{Lemma}

\newtheorem{proposition}[theorem]{Proposition}
\newtheorem{corollary}[theorem]{Corollary}
\theoremstyle{definition}

\newtheorem{remark}[theorem]{Remark}

\usepackage{accents}
\newlength{\dhatheight}

\let\PLAINthebibliography\thebibliography
\renewcommand\thebibliography[1]{
  \PLAINthebibliography{#1}
  \setlength{\parskip}{0.5pt}
  \setlength{\itemsep}{0.5pt plus .3ex}
}

\newcommand{\ten}{1\!0}

\newcommand{\shape}{
  \raisebox{1pt}{\rm\normalfont\textesh}
}

\newcommand{\differential}{\mathrm{d}}

\newcommand\bos[1]{\mathstrut\mkern2.5mu#1\mkern-14mu\raise1.7ex%
  \hbox{$\scriptstyle\rightsquigarrow$}}

\newcommand\bosonic[1]{\mathstrut\mkern2.5mu#1\mkern-14mu\raise1.7ex%
  \hbox{$\scriptstyle\rightsquigarrow$}}

\usepackage[cmtip,all]{xy}
\newcommand{\longsquiggly}{\xymatrix{{}\ar@{~>}[r]&{}}}

\usepackage{bbm} 

\usepackage[new]{old-arrows}   

\newcommand{\FR}{\mathbb{R}}

\newcommand{\dd}{\mathrm{d}}

\newcommand{\even}{\mathrm{even}}

\newcommand{\odd}{\mathrm{odd}}

\newcommand{\om}{\omega}

\newcommand{\CX}{\mathcal{X}}

\begin{document}

\setlength{\abovedisplayskip}{3pt}
\setlength{\belowdisplayskip}{3pt}
\setlength{\abovedisplayshortskip}{-3pt}
\setlength{\belowdisplayshortskip}{3pt}

\title{
Flux Quantization on 10D Type IIA Superspace \\
via Cyclification from 11D }

\author{
  Grigorios Giotopoulos${}^{\ast}$,
  \;\;
  Hisham Sati${}^{\ast \dagger}$,
  \;\;
}

\maketitle

\thispagestyle{empty}

\begin{abstract}  We produce the dimensional reduction to 10D IIA supergravity (SuGra), via cyclification, of 
  the remarkable result that full 11D SuGra is put on shell just by imposing the duality-symmetric Bianchi identities on C-field super-flux densities over supertorsion-free superspace.

\smallskip 
  Generally, we highlight that when duality-symmetric superspace Bianchi identities are characterized by Whitehead bracket $L_\infty$-algebras $\mathfrak{l}\mathcal{A}$ of a classifying space $\mathcal{A}$,  
  their dimensional reduction is characterized by the cyclic loop space $\mathrm{Cyc}(\mathcal{A})$. We promote this to a general mechanism of dimensional reduction on super-spacetime, compatible with the global (infrared) completion of supergravity theories by flux quantization in non-abelian cohomology with coefficients in $\mathcal{A}$ and $\mathrm{Cyc}(\mathcal{A})$, respectively.

\smallskip 
  In the case of 11D SuGra, the characteristic $L_\infty$-algebra is $\mathfrak{l}S^4$ and hence we obtain that full on-shell 10D IIA SuGra is equivalent to $\mathfrak{l}\mathrm{Cyc}(S^4)$-Bianchi imposed identities on NS/RR super-flux densities over supertorsion-free 10D super-spacetime. This implies that any space which is $\mathbb{R}$-rationally equivalent to $\mathrm{Cyc}(S^4)$ classifies an admissible flux quantization law, which provides a global completion of 10D IIA SuGra that admits oxidation to 11D.
\end{abstract}

\vspace{1cm}

\begin{center}
\begin{minipage}{15cm}
\tableofcontents
\end{minipage}
\end{center}

\medskip

\vfill

\hrule
\vspace{5pt}

{
\footnotesize
\noindent
\def\arraystretch{1}
\tabcolsep=0pt
\begin{tabular}{ll}
${}^*$\,
&
Mathematics, Division of Science; and
\\
&
Center for Quantum and Topological Systems,
\\
&
NYUAD Research Institute,
\\
&
New York University Abu Dhabi, UAE.  
\end{tabular}
\hfill
\href{https://ncatlab.org/nlab/show/Center+for+Quantum+and+Topological+Systems}{
\adjustbox{raise=-15pt}{
\includegraphics[width=3cm]{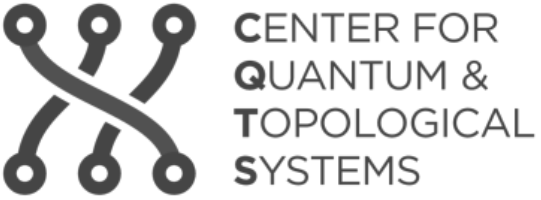}
}
}

\vspace{1mm} 
\noindent ${}^\dagger$\ Courant Institute for Mathematical Sciences, NYU, NY

\vspace{.2cm}

\noindent
The authors acknowledge the support by {\it Tamkeen} under the 
{\it NYU Abu Dhabi Research Institute grant} {\tt CG008}.
}

\newpage

\section{Introduction \& Overview}\label{IntroductionSection}

The existence of supersymmetry in eleven dimensions \cite{Nahm78} makes possible the supersymmetric (Lagrangian) formulation 
of 11-dimensional supergravity \cite{CJS78}(reviewed in \cite{MiemiecSchnakenburg06}). 
The classical local (coordinate)  dimensional reduction on a circle leads to 10-dimensional type IIA supergravity  \cite{CW84}\cite{GP84} \cite{HN85}. 
Globally, this reduction has also been accomplished at the level of non-trivial bundles and in a coordinate-independent manner in the second-order (metric) formulation \cite{FOS03}\cite{MaS04} (see also 
\cite{Jose01}). Nevertheless, the dimensional reduction of the global nature of the higher gauge fields of 11D supergravity has not been fully taken into account. As we will see, 
the proper way to discuss the latter, in a manner compatible with the full dynamical supergravity field equations, is via the formalism of ``superspace supergravity''.

\smallskip 
Superspace formulations of 11D supergravity proceed by ``solving the Bianchi identities'' \cite{CF80}\cite{BrinkHowe80} (see also \cite{Howe97}\cite{CGNT05}\cite{EEC12}) on 11-dimensional supermanifolds, under the demand of certain further conditions on the super-coframe expansions of the torsion and the superized flux fields. Such 
superspace formulations allow, in particular, \cite[\S III.8.5]{CDF91} for an \textit{on-shell duality-symmetric form} of the theory 
(cf. \cite{DF82}\cite{BBS98}). A rigorous account along these lines, phrased in a modern geometrical language suitable for the purpose of flux quantization, along with
systematic derivations together with solid proofs, is given in \cite{GSS24-SuGra}.
Dimensional reduction of superspace formulations of the theory to 10D has been pursued at the local coordinate-level in  \cite{HS05}\cite{DFGT} using standard Kaluza--Klein techniques in superspace, but has yet to be cleanly performed for arbitrary $S^1$-principal bundle topologies, nor via the full duality-symmetric form of 11D supergravity, which is necessary for our non-trivial applications of flux quantization. The latter is one of our achievements here.

\smallskip  
We note that there also exist different accounts of IIA superspace formulations via (the now ``K-theoretic'') Bianchi identities by imposing constraints on the (full set of) type IIA NS/RR fluxes such as \cite{CGNSW97}, following \cite{CGO87}, and via a different set of independent fields and their Bianchi identities \cite{NO11} (but these are not necessarily derived from eleven dimensions). 
We stress that some of the superspace formulations happen to be \textit{duality-symmetric}, a property crucial for our purposes both for 11D 
 and type 10D IIA supergravities, and furthermore always manifestly \textit{on-shell}. We will later provide a detailed comparison of our explicit formulation with all existing ones (p. 19). An approach towards off-shell duality-symmetric Lagrangian formulations of these theories is given in  \cite{BBS98}\cite{BNS04}, but this will not be relevant for our goal here.

\smallskip 

In this paper, thus, we start from the duality-symmetric superspace formulation of (on-shell) 11D supergravity via the sole (and equivalent) demand that the superized $(G_4^s,G_7^s)$-fluxes constitute an $\mathfrak{l}S^4$-cocycle \cite{GSS24-SuGra}. This formulation is crucial, as it dictates the allowable flux quantizations of 11D supergravity fluxes to be in a \textit{non-abelian cohomology} $\mathcal{A}$ of rational homotopy type that of the 4-sphere $\mathfrak{l}\mathcal{A} = \mathfrak{l} S^4$ (\cite{Sati13}\cite{FSS20-H}\cite{FSSHopf}\cite{FSS22Twistorial}). We show that dimensional reduction along the fiber of arbitrary $S^1$-principal bundle super-spacetimes implies a superspace formulation of (on-shell) 10D IIA supergravity via the sole (and equivalent)\footnote{Just as with the case of 11D supergravity (cf. \cite{GSS24-SuGra}), in all existing IIA superspace formulations recorded in published material, only one directional implication is presented -- rather than the bi-directional equivalence \eqref{10DTabledStatement} that we prove.}  demand that the reduced duality-symmetric NS/RR $(F_2^s, F_4^s, F_6^s, H_3^s, H_7^s)$-fluxes are valued in the cyclification of the rational 4-sphere, $\mathrm{cyc}(\mathfrak{l}S^4)$. The latter rational classifying space, together with its natural compatibility with the corresponding classifying space from 11D, and the full dynamical IIA field equations, strongly prompts (and allows) us to conjecture that IIA supergravity fluxes are to be quantized in the corresponding admissible cyclified (and necessarily) \textit{non-abelian} cohomology theory $\mathrm{Cyc}(\mathcal{A})$. This is to be contrasted with the famous hypothesis of flux quantization of IIA supergravity in \textit{abelian} twisted K-theory 
which becomes possible only when one (simply) ignores the dual NS $H_7$-flux and its \textit{non-linear} Bianchi identity (see \cite{BMSS19}). 

\medskip 
\noindent {\textbf{Mathematical background and conventions.} We will mainly follow the conventions and notation of \cite{GSS24-SuGra}\cite{GSS24-TDuality}. The necessary mathematical background and context is extensively reviewed therein, where further original literature is cited, so we will refrain from doing this again in this short text.}
\medskip 

\noindent \textbf{The natural Cartan geometric superspace formulation of 11D supergravity.}
Using modernized mathematical language, the result of \cite{GSS24-SuGra} may be succinctly summarized as the following \textit{geometric fact}. The first step is to note that on the (flat) super-Minkowski spacetime 
$
\FR^{1,10\vert \mathbf{32}}  
$
equipped with its super-Lie algebra structure, encoded in its dual Chevalley--Eilenberg algebra 
\begin{equation}\label{SuperMinkowskiCE}
\mathrm{CE}(\FR^{1,10\vert \mathbf{32}})\,  = \, \FR_d \big[ (e^a)_{a=0}^{10}, (\psi^\beta)_{\beta=1}^{32} \big]\, \big/\, ( {\scriptstyle \dd e^a = (\overline{\psi} \, \Gamma^a \, \psi), \, \dd \psi^\beta =0} ) \, ,
\end{equation} 
there exists the following special pair of $4$- and $7$-cochains $(G_4^0, G_7^0)$ given by
\begin{align}\label{FlatG4G7}
G_4^0 \, &= \, { \color{darkorange}
      \tfrac{1}{2}
      \big(\,
        \overline{\psi}
        \,\Gamma_{a_1 a_2}\, \psi
      \big)
      e^{a_1} e^{a_2}
      } \\
G_7^0  \, &= \,   {
      \color{darkorange}
      \tfrac{1}{5!}
      \big(\,
        \overline{\psi}
        \,\Gamma_{a_1 \cdots a_5}\, \psi 
      \big)
      e^{a_1} \cdots e^{a_5}
      }\nonumber \, .  
\end{align}
Crucially, these form an $\mathfrak{l}S^4$-cocycle ($\equiv$  map of super-$L_\infty$-algebras)
$$
\FR^{1,10\vert \mathbf{32}}  \xrightarrow{\quad \quad } \mathfrak{l} S^4, 
$$
i.e., dually a map of super-commutative differential graded algebras (sCDGA), $\mathrm{CE}(\mathfrak{l}S^4)\longrightarrow \mathrm{CE}(\FR^{1,10\vert \mathbf{32}})$, in that they satisfy the familiar ``duality-symmetric'' Bianchi identities
\begin{align*}
\dd G_4^0 \, &= \, 0 \\
\dd G_7^0 \, &= \, \tfrac{1}{2}G_4^0 \, G_4^0 \, . 
\end{align*}

The next step is to view the relation \eqref{FlatG4G7} as a \textit{tangent space-wise} identity on any \textit{super-spacetime}  \cite[Def 2.74]{GSS24-SuGra}  supermanifold modelled on $\FR^{1,10\vert \mathbf{32}}$
$$
\big(Y^{11}, (e,\psi,\omega)\big)\, ,
$$
with $(e,\psi,\omega)$ forming a \textit{Cartan connection}\footnote{In particular, this structure guarantees the consistency of the former ``tangent-space-wise'' interpretation, since the super-coframe yields an isomorphism of super-vector spaces $(e,\psi)|_p : T_p Y^{11} \xrightarrow{\sim} \FR^{1,10\vert \mathbf{32}}$ for each $p\in \bosonic{Y}^{11}\hookrightarrow Y$.} for the super-group\footnote{See \cite{GSS24-MGroup} for a rigorous and modernized mathematical treatment of such super-groups.} 
inclusion $\mathrm{Spin}(1,10)\hookrightarrow \mathrm{ISO}(\FR^{1,10\vert \mathbf{32}}).$ 
Then, the statement of \cite[Thm. 3.1]{GSS24-SuGra} is that \textit{full on-shell 11D supergravity configurations} are \textit{precisely} the {\color{darkblue} curved global extensions} of these {\color{darkorange}flat  tangent space} $\mathfrak{l}S^4$-cocycles. That is, exactly those $(11\vert \mathbf{32})$-dimensional super-spacetimes that carry global super-fluxes $(G_4^s, G_7^s)$ of the form\footnote{While in \eqref{FlatG4G7} the symbols $e^a$ and $\psi^\beta$ denote generators of CE$(\FR^{1,10\vert \bf{32}})$, here they are components of the super-coframe $(e,\psi)$, hence generating $\Omega^\bullet_\mathrm{dR}(Y^{11})$. The two notions coincide upon restricting the latter to each super-tangent space.}   
\begin{align}\label{CurvedG4G7}
G_4^s \, &= \, {
        \color{darkblue}
        \tfrac{1}{4!}
        (G_4)_{a_1 \cdots a_4}
        \,
        e^{a_1} \cdots e^{a_4}
      } + { \color{darkorange}
      \tfrac{1}{2}
      \big(\,
        \overline{\psi}
        \,\Gamma_{a_1 a_2}\, \psi
      \big)
      e^{a_1} e^{a_2}
      } \\
G_7^s  \, &= \,  {
      \color{darkblue}
      \tfrac{1}{7!}
      (G_7)_{a_1 \cdots a_7}
      \,
      e^{a_1} \cdots e^{a_7}
      } +  {
      \color{darkorange}
      \tfrac{1}{5!}
      \big(\,
        \overline{\psi}
        \,\Gamma_{a_1 \cdots a_5}\, \psi
      \big)
      e^{a_1} \cdots e^{a_5}
      }\nonumber  
\end{align}
which still form a (rational) $4$-sphere cocycle, i.e., now as a map of super-$L_\infty$-\textit{algebroids}
$$
T Y^{11}  \xrightarrow{\quad \quad } \mathfrak{l} S^4, 
$$
in that they still satisfy the corresponding duality-symmetric Bianchi (now differential) identities 
\begin{align*}
\dd G_4^s \, &= \, 0 \\
\dd G_7^s \, &= \, \tfrac{1}{2}G_4^s \wedge G_4^s \,  
\end{align*}
in $\Omega^\bullet_\mathrm{dR}(Y^{11}) \cong \mathrm{CE}(TY^{11})$. Yet equivalently, this means they are modulated by a map into the classifying space of closed $\mathfrak{l}S^4$-valued forms
$$
(G_4^s, G_7^s) \, : \, Y^{11} \xrightarrow{\quad \quad} \Omega^{1}_{\mathrm{dR}}(-;\, \mathfrak{l}S^4)_{\mathrm{clsd}} \, .
$$

Here the terms in {\color{darkorange}orange} are to be thought of as the \textit{canonical} supersymmetric contributions arising from the {\color{darkorange}flat but super-}tangent space structure, while the terms in {\color{darkblue}blue} are the flux densities arising from the underlying {\color{darkblue}bosonic but curved} manifolds. The combination of the two on super-spacetime supermanifolds $Y$, under the demand that they still satisfy the $\mathfrak{l}S^4$-Bianchi identities, enforces exactly the equations of motion of 11D supergravity (and importantly, no further conditions!).   

\medskip 
\noindent \textbf{11D supergravity flux quantization in non-abelian cohomology via the 4-sphere.}
Other than yielding an alternative, concise and natural geometric formulation of 11D supergravity, the simultaneous appearance of the rational $4$-sphere as a classifying space \textit{and} its compatibility with the \textit{full set of field equations}, including the Einstein coframe (``metric'') dynamics, Hodge duality constraint \footnote{For higher gauge field theories on \textit{fixed metric} backgrounds, the Hodge-duality constraint may be safely ignored by passing to the Cauchy data 
\cite{SS23FQ}. However, this does not cover super-gravity theories where the metric is dynamical.} and fermionic contributions \footnote{In particular, it is worth noting that on the underlying bosonic spacetime $\bosonic{Y}^{11}\hookrightarrow Y^{11}$, the fermionic contributions actually spoil the form of the purely bosonic duality-symmetric Bianchi identity (see, e.g., \cite[(2.2.7), (2.2.14)]{DNP86}).},
strongly support the suggestion that the global ``infrared'' completion of the (higher gauge) $C$-field should be in a certain \textit{non-abelian (differential) cohomology} theory  $\mathcal{A}$ \cite{FSS23Char}. It is worth stressing here that such a cohomology is topologically much richer than ordinary \textit{abelian} differential cohomology, which would be the indication from considering the sole (non-duality symmetric) Bianchi identity $\dd G_4 = 0$.  

\smallskip 
With no further physical input and/or expectations (e.g., from the consistent coupling of M-branes to such backgrounds), the only  mathematical requirement of such a non-abelian cohomology theory is that its classifying space has as rational homotopy type that of the $4$-sphere 
$$
\mathfrak{l} \mathcal{A} \,  \cong \,  \mathfrak{l} S^4 \, ,
$$ encoded precisely in the \textit{duality-symmetric} Bianchi identity of the $C$-field. It follows, then, that a natural and ``minimal'' hypothesis for this choice is the \textit{actual} topological $4$-sphere $\mathcal{A}= S^4$, i.e., Cohomotopy cohomology theory, which has come under the name ``Hypothesis H'' (originally proposed in \cite[\S 2.5]{Sati13}, formalized and checked in \cite{FSS20-H}\cite{FSSHopf}\cite{FSS22Twistorial} to reproduce various expectations from the literature, and reviewed in \cite[\S 12]{FSS23Char}\cite[\S 4.3]{SS24Flux}).  For more details on this aspect of flux quantization on $11D$ \textit{super-spacetime}, see \cite{GSS24-SuGra}, and for that of the corresponding $M5$-brane coupled to such backgrounds see \cite{GSS-M5Brane}.

\medskip
\noindent \textbf{The natural dimensionally-reduced Cartan geometric superspace formulation of IIA 10D SuGra.}
It was observed in \cite{FSS15-WZW}\cite{FSS15-nAlg}\cite{FSS17}, reviewed and further modernized in \cite[\S 2.2]{BMSS19}\cite[\S 2.2]{GSS24-TDuality}), that the algorithm of dimensional reduction of flux-forms at the (super-)tangent space level has a very natural interpretation in terms of an adjunction ($\equiv$ ``duality'') in the realm of super-$L_\infty$-algebras. 

\smallskip 
In particular, for the case of flat $(11\vert \mathbf{32})$-dimensional super-Minkowski spacetime being a central extension of the $(10\vert \bf{16}\oplus \overline{\bf{16}})$-dimensional IIA super-Minkowski spacetime (see \cite[\S3.1]{GSS24-TDuality} for details)
$$
\FR^{1,10\vert \mathbf{32}} \xrightarrow{\quad \quad } \FR^{1,9\vert \mathbf{16} \oplus \overline{\mathbf{16}}} {\color{gray} \xrightarrow{\;\; \scriptscriptstyle{c_1^M= (\overline{\psi}\Gamma^{\ten}   \psi)}\;\; } b \FR}
$$
given dually by the sCDGA inclusion\footnote{In this formulation, the top dimensional $\Gamma$-matrix is identified simply as the linear operator $\Gamma^{10} = P- \overline {P}$, for $P$ and $\overline{P}$ being the two projectors on each of the (Spin$(1,9)$ irrep.) components of $\mathbf{16}\oplus \overline{\mathbf{16}}$, respectively.}
\begin{align*}
\mathrm{CE}\big(\FR^{1,9\vert \mathbf{16} \oplus \overline{\mathbf{16}}}\big) &\xhookrightarrow{\qquad} \mathrm{CE}\big({\FR^{1,10\vert \mathbf{32}}}\big) 
\, \cong \,  \mathrm{CE}\big(\FR^{1,9\vert \mathbf{16} \oplus \overline{\mathbf{16}}}\big)[e^{\ten}] \, \big/ \, \scriptstyle{ \dd e^{\ten} = (\overline{\psi} \Gamma^{\ten} \psi) }\,,
\end{align*} 
there exists a (natural) bijection of morphisms of super-$L_\infty$-algebras\footnote{
This Hom-bijection yields an actual adjunction between the central extension (``oxidation'') functor $\mathrm{oxi}: sL_\infty\mbox{-}\mathrm{Alg}_{/ b\FR} \rightarrow sL_\infty\mbox{-}\mathrm{Alg} $ and the cyclification functor $\mathrm{cyc}: sL_\infty\mbox{-}\mathrm{Alg}\rightarrow sL_\infty\mbox{-}\mathrm{Alg}_{/ b\FR}$. For expositional purposes, here we consider only the cases of interest at hand for source and target.}
\begin{align}\label{FlatExtCycBijection}
\mathrm{Hom}_{sL_\infty\mbox{-}\mathrm{Alg}}\big(\FR^{1,10\vert \mathbf{32}}, \, \mathfrak{l}S^4\big) \;\; \cong \;\; \mathrm{Hom}_{sL_\infty\mbox{-}\mathrm{Alg}_{/ b\FR}}\big(\FR^{1,9\vert \mathbf{16}\oplus\overline{\bf{16}}}, \, \mathrm{cyc}(\mathfrak{l}S^4)\big) \, . 
\end{align}
Applied on the {\color{darkorange}flat tangent-space} $\mathfrak{l}S^4$-cocycle from \eqref{FlatG4G7}
$$
  \begin{tikzcd}[
  ]
    \mathbb{R}^{1,10\,\vert\,\mathbf{32}}
    \ar[
      rr,
      "{
        (G_4^0,\, G_7^0)
      }"
    ]
    &&
    \mathfrak{l}S^4
  \end{tikzcd}
  \hspace{1cm}
    \longleftrightarrow
  \hspace{1cm}
  \begin{tikzcd}[
      row sep=-1pt, column sep=huge
  ]
    \mathbb{R}^{
      1,9\,\vert\, 
      \mathbf{16} 
        \oplus 
      \overline{\mathbf{16}}
    }
    \ar[
      rr,
      "{
        \mathrm{rdc}_{
          \scalebox{.7}{$c_1^M$}
        }
          (G^0_4,G^0_7)
      }"
    ]
    \ar[
      dr,
      "{
        (\hspace{1pt}
          \overline{\psi}
          \,\Gamma_{\!\ten}\,
          \psi) 
      }"{sloped, swap}
    ]
    &&
    \mathrm{cyc}(\mathfrak{l}S^4) \,,
    \ar[
      dl,
      "{ \omega_2 }"
    ]
    \\
    &
    b\mathbb{R}
  \end{tikzcd}
$$
this reduction by ``\textit{cyclification}'' acts explicitly as
\vspace{-2mm} 
\begin{equation}
  \label{FlatNS/RRCocycles}
  \hspace{-5mm}
  \begin{tikzcd}[sep=-1pt]
    \mathbb{R}^{1,10\,\vert\,\mathbf{32}}
    \ar[
      rr,
      "{
        (G^0_4,\, G^0_7)
      }"
    ]
    &&
    \mathfrak{l}S^4
    \\
    {\color{darkorange} \tfrac{1}{2}
    \big(\hspace{1pt}
      \overline{\psi}
      \,\Gamma_{a_1 a_2}\,
      \psi
    \big)
    e^{a_1} e^{a_2} }
    &\longmapsfrom&
    g_4
    \\
    {\color{darkorange} \tfrac{1}{5!}
    \big(\hspace{1pt}
      \overline{\psi}
      \,\Gamma_{a_1 \cdots a_5}\,
      \psi
    \big)
    e^{a_1} \cdots e^{a_5} }
    &\longmapsfrom&
    g_7
  \end{tikzcd}  
  \hspace{.25cm}
  \leftrightsquigarrow
  \hspace{-4mm}
  \begin{tikzcd}[
    row sep=-2pt
    ,column sep=25pt
    ,/tikz/column 1/.append style={anchor=base east}
  ]
    \mathbb{R}^{
      1,9\,\vert\,
      \mathbf{16} 
        \oplus 
      \overline{\mathbf{16}}
    }
    \ar[
      rr,
      "{\small 
          \begin{array}{c}
            \mathrm{rdc}_{
              \scalebox{.7}{$c_1^M$}
            }
              (G^0_4, G_7^0) 
            \;=
            \\
            (F_2^0, F_4^0, F_6^0, H_3^0, H_7^0)
          \end{array}
      }"
    ]
    &&
    \mathrm{cyc}\big(
      \mathfrak{l}S^4
    \big)
    \\
    {\color{darkorange}  \big(\hspace{1pt}
      \overline{\psi}
      \,\Gamma_{\!\ten}
      \psi
    \big) }
    \;=:\;
    \mathrm{F}_2^0
    &\longmapsfrom&
    \omega_2
    \\
     {\color{darkorange}  \tfrac{1}{2}\big(\hspace{1pt}
      \overline{\psi}
      \,\Gamma_{a_1 a_2}\,
      \psi
    \big)
    e^{a_1} e^{a_2} }
    \;=:\;
    \mathrm{F}_4^0
    &\longmapsfrom&
    g_4
    \\
    {\color{darkorange} -  \big(\hspace{1pt}
      \overline{\psi}
      \,
      \Gamma_{\! \ten}
      \Gamma_{a}
      \,
      \psi
    \big)
    e^a ) }
    \;=:\;
    H_3^0
    &\longmapsfrom&
    \mathrm{s}g_4
    \\
    {\color{darkorange}  \tfrac{1}{5!}
    \big(\hspace{1pt}
      \overline{\psi}
      \,\Gamma_{a_1 \cdots a_5}\,
      \psi
    \big)
    e^{a_1} \cdots e^{a_5} }
    \;=:\;
    H^0_7
    &\longmapsfrom&
    g_7
    \\
   {\color{darkorange}  -
    \tfrac{1}{4!}
    \big(\hspace{1pt}
      \overline{\psi}
      \,\Gamma_{\!\ten}
      \Gamma_{a_1 \cdots a_4}
      \,
      \psi
    \big)
    e^{a_1} \cdots e^{a_4} }
    \;=:\;
    F_6^0
    &\longmapsfrom&
    \mathrm{s}g_7
  \end{tikzcd}
\end{equation}
that is, with the reduced cochains satisfying the corresponding $\mathrm{cyc}(\mathfrak{l}S^4)$-Bianchi identities:\footnote{The sign in the third equation can be absorbed in the definition of $F_6^0$, hence bringing it closer to the familiar (rational) twisted K-theory formulas. However for us, both from our dimensional reduction perspective and the symmetry of the formulas in \eqref{FlatNS/RRCocycles}, it is mathematically natural to retain it.}
\vspace{2mm} 
\begin{equation}
  \label{DimReducedBianchiIdentities}
  \left.
  \def\arraystretch{1.3}
  \def\arraycolsep{4pt}
  \begin{array}{lcl}
    \mathrm{d}\, G_4^0 &=& 0
    \\
    \mathrm{d}\, G_7^0 &=&
    \tfrac{1}{2}\,
    G_4^0 \, G_4^0
  \end{array}
  \!\!\right\}
  \hspace{.5cm}
  \leftrightsquigarrow
  \hspace{.5cm}
  \left\{\!\!
  \def\arraystretch{1.1}
  \def\arraycolsep{4pt}
  \begin{array}{lcl}
    \mathrm{d}\, F_2^0
    &=&
    0
    \\
    \mathrm{d}\, F_4^0
    &=&
    H_3^0 \, F_2^0
    \\
    \mathrm{d}\, F_6^0
    &=&
    -\,  H_3^0 \, F_4^0
    \\
    \mathrm{d}\, H_3^0
    &=&
    0
    \\
    \mathrm{d}\, H_7^0
    &=&
    \tfrac{1}{2}
    \,
    F_4^0 \, F_4^0
    \,+\,
    F_2^0 \, F_6^0\;.
  \end{array}
  \right.
\end{equation}

Now, since the curved global 11D super-spacetime \eqref{CurvedG4G7} extensions of the tangent-space $\mathfrak{l}S^4$-cocycles \eqref{FlatG4G7} are nothing but the full 11D supergravity configurations (\cite[Thm. 3.1]{GSS24-SuGra}), the above tangent space-wise reduction \eqref{FlatExtCycBijection} suggests that -- should it have an appropriate global extension in the non-trivial supermanifold setting --
it would automatically produce a 10D IIA super-spacetime supergravity formulation, since the latter IIA supergravity is known to be the dimensional reduction of 11D supergravity. Indeed, our Prop. \ref{NonTrivialCircleBundleCyclification/Oxidation} realizes the global supermanifold extension of \eqref{FlatExtCycBijection} precisely in the physical scenario at hand: The dimensional reduction of collections of ($S^1$-symmetric) flux-form fields on a  principal $S^1$-bundle $Y$ over a base supermanifold $X$, via the use of $U(1)$-connection $e'$.

\smallskip 
In the case of SuGra super-spacetimes, the necessary extra datum of the U(1)-connection $e'$ arises from a special partial gauge fixing \cite[Cor. 3.16]{Gi26} of the underlying $S^1$-symmetric gravitational field content, so that in particular $e^{10}_Y = \Phi\, e'$. This gauge fixing produces the expected (superized) dilaton field $\Phi$, but also its partner dilatino field $\lambda$ via $\psi_Y = \psi_X + \lambda \, e'$, each encoding the ``winding'' of the bosonic $e_Y$ and fermionic coframe $\psi_Y$ along the bosonic ``M-theoretic'' $S^1$-fiber, respectively. A careful application of Prop. \ref{NonTrivialCircleBundleCyclification/Oxidation} on the space of the $S^1$-symmetric $11D$ super-spacetimes $\big(Y^{11}, (e_Y,\psi_Y,\omega_Y)\big)$ supplied with $S^1$-symmetric $\mathfrak{l}S^4$-super-fluxes $(G_4^s, G_7^s)$ of the form \eqref{CurvedG4G7}, and expanded in the mentioned partial gauge fixing, yields the following {\color{darkblue}curved global extensions} of the {\color{darkorange} flat tangent space} $\mathrm{cyc}(\mathfrak{l}S^4)$-cocycles \eqref{FlatNS/RRCocycles}, now on the corresponding base $(10\vert \mathbf{16}\oplus \overline{\mathbf{16}})$-dimensional super-spacetime $X^{10}$   
\begin{equation}
\label{CurvedNS/RRCocycles}
  \def\arraystretch{1.5}
  \begin{array}{l}
    F^s_2
    \;\;:=\;\;
    {\color{darkblue} \tfrac{1}{2}(F_2)_{a_1 a_2} \, e^{a_1} e^{a_2} }\;\, + \, \, {\color{olive} \frac{2}{\Phi^2} (\overline{\psi}  \Gamma_{a} \lambda)\, e^a }\,\, +\;\, {\color{olive} \frac{1}{\Phi}} \, {\color{darkorange}  (\overline{\psi} \Gamma_{\ten} \psi) }
    \\
        H^s_3
    \;\;:=\;\;  {\color{darkblue} \tfrac{1}{3!}(H_{3})_{a_1 a_2 a_3} \, e^{a_1}e^{a_2}e^{a_3} }\,\, + \, \, 
   {\color{olive} (
    \overline{\psi}
    \Gamma_{a_1 a_2}
    \lambda )\,
    e^{a_1}  e^{a_2} } \, \,  \;\,- \;\, {\color{olive} \Phi} \, {\color{darkorange}  (\overline{\psi} \Gamma_{\ten a} \psi)\, e^a }
    \\
        F^s_4
    \;\;:=\;\;
    {\color{darkblue}  \frac{1}{4!}(F_4)_{a_1 \cdots a_4} \, e^{a_1}\cdots e^{a_4} }\;\, 
    \hspace{3.7cm} 
    +\;\, {\color{darkorange} \tfrac{1}{2}(\overline{\psi} \Gamma_{a_1 a_2} \psi)\,  e^{a_1} e^{a_2} }
    \\
        F^s_6
    \;\;:=\;\; 
    {\color{darkblue}  \frac{1}{6!}(F_6)_{a_1 \cdots a_6}\,  e^{a_1}\cdots e^{a_6} }\;\,+\;\, {\color{olive}
    \tfrac{2}{5!}
    (
    \overline{\psi}
    \Gamma_{a_1 \cdots a_5}
    \lambda
    ) \, 
    e^{a_1} \cdots e^{a_5} } \;\,-\;\, {\color{olive} \tfrac{\Phi}{4!} } \,  {\color{darkorange}(\overline{\psi} \Gamma_{\ten \, a_1 \cdots a_4} \psi) \,  e^{a_1}\cdots e^{a_4}  }
    \\
    H^s_7
    \;\;:=\;\; 
    {\color{darkblue}  \tfrac{1}{7!}
    (H_7)_{a_1 \cdots a_7}\, 
    e^{a_1} \cdots e^{a_7} }
    \;\,   
    \hspace{4.55cm} 
    + \;\, {\color{darkorange} 
    \tfrac{1}{5!}
    (
    \overline{\psi}
    \Gamma_{a_1 \cdots a_5}
    \psi
    ) \, 
    e^{a_1} \cdots e^{a_5}  }\, .
  \end{array}
\end{equation}
Here, we also have $\dd e' = \pi^* F_2^s$ and  $\dd \Phi \, = \, \partial_a\Phi\, e^a \, + \, 2 \big(\,\overline{\psi} 
    \,\Gamma^{\ten}\,
    \lambda
    \big)$,
with the additional terms in {\color{olive} olive} arising due to the {\color{olive}winding} of the 11D super-coframe along the ``shrinking'' {\color{olive}M-theoretic circle}.

\smallskip 
Our main Theorem \ref{10dIIASugraEoMFromSuperFluxBianchiIdentity} then guarantees that these global extensions correspond \textit{precisely} to the \textit{full on-shell 10D IIA supergravity configurations} (and nothing more!), by virtue of these forming a (rational) cyclified 4-sphere cocycle
$$
T X^{10} \xrightarrow{\quad \quad } \mathrm{cyc}(\mathfrak{l}S^4)\, ,
$$
in that they still satisfy the duality-symmetric NS/RR Bianchi identities \eqref{DimReducedBianchiIdentities}, now in $\Omega^\bullet_\mathrm{dR}(X) \cong \mathrm{CE}(TX)$
\begin{align*}
\mathrm{d}\, F_2^s
    &=
    0
    \\
    \mathrm{d}\, F_4^s
    &=
    H_3^s \, F_2^s
    \\
\mathrm{d}\, F_6^s
    &=
    -\,  H_3^s \, F_4^s
    \\    \mathrm{d}\, H_3^s
    &=
    0
    \\
\mathrm{d}\, H_7^s
    &=
    \tfrac{1}{2}
    \,
    F_4^s \, F_4^s
    \,+\,
    F_2^s \, F_6^s\;.
\end{align*}
Equivalently, this says that the superized NS/RR-fluxes are modulated by a map into the classifying space of closed $\mathrm{cyc}(\mathfrak{l}S^4)$-valued forms
$$
(F_2^s, H_3^s, F_4^s, F_6^s, H_7^s) \, : \, X^{10} \xrightarrow{\quad \quad} \Omega^{1}_{\mathrm{dR}}\big(-;\,\mathrm{cyc}( \mathfrak{l}S^4)\big)_{\mathrm{clsd}} \, .
$$

\smallskip 
Summarizing, by construction, the superized duality-symmetric IIA NS/RR-field fluxes \eqref{CurvedNS/RRCocycles} satisfy the same principle as that observed in 11D for the superized duality-symmetric C-field flux $(G_4^s, G_7^s)$, namely they unify the structures of:  
\begin{itemize} 
\item[\bf{(i)}] NS/RR flux densities on {\color{darkblue} bosonic but non-trivial} underlying manifolds $\bosonic{X}\hookrightarrow X$, which satisfy the corresponding duality-symmetric NS/RR cyc$(\mathfrak{l}S^4)$-Bianchi identities;

\item[\bf(ii)] the canonical supersymmetric contributions arising on {\color{darkorange} flat but super-}tangent spaces $\FR^{1,9\vert \bf{16}\oplus\overline{\bf{16}}}$ satisfying the analogous identities; 

\item[\bf{(iii)}] 
now along with certain additional subtle ``{\color{olive} mixed components extensions}'', reflecting the fact that the total global NS/RR extensions  are really reductions of 11D fluxes, whose super-coframe expansions {\color{olive}wind non-trivially} along the (bosonic) $S^1$-fiber.
\end{itemize} 

\medskip 
\noindent \textbf{10D IIA supergravity flux quantization in cyclified non-abelian cohomology.}
Similar to the case of 11D supergravity, the above provides a concise and natural geometric formulation of 10D IIA SuGra. Additionally, however, the simultaneous {\bf (i)} appearance of the rational classifying $L_\infty$-algebra $\mathrm{cyc}(\mathfrak{l}S^4)$, {\bf (ii)} its compatibility with the \textit{full IIA field equations}, and {\bf (iii)} its compatibility with the 11D rational classifying space $\mathfrak{l}S^4$ via oxidation/cyclification, now strongly suggests that the global completion of the IIA flux fields, in the guise of duality-symmetric NS/RR collection $(F_2,H_3, F_4, F_6, H_7)$, should also be in a certain \textit{non-abelian} cohomology $\mathcal{B}^{\rm IIA}$.

\smallskip 
However, since rational cyclification is nothing but the rationalization of the topological cyclification,  $\mathrm{Cyc}(\mathcal{A}) := \mathrm{Map}(S^1, \mathcal{A})/\!/ S^1$, of any space \cite[Thm. A]{VPB85} (amplified in our current context in \cite[Prop. 3.2]{FSS17}\cite{BMSS19}\cite{SV3}) 
$$
\mathrm{cyc}(\mathfrak{l}\mathcal{A}) \, \cong \, \mathfrak{l}\mathrm{Cyc}(\mathcal{A}) \, , $$ 
this implies that the non-abelian IIA flux quantization may be chosen to be \textit{compatible} with that of 11D supergravity, upon dimensional reduction via cyclification/oxidation! Of course, this is precisely what one would expect physically, that is, if the full \textit{globally completed} IIA supergravity is really to be a limit of the full \textit{globally completed} 11D supergravity (and not only in local coordinate patches). In other words, if $\mathcal{A}$ is an admissible flux quantization law for 11D supergravity, so that $\mathfrak{l}\mathcal{A} \cong \mathfrak{l}S^4$, then the induced admissible flux quantization IIA supergravity is given by
$$
\mathcal{B}^{\rm IIA} \, = \, \mathrm{Cyc}(\mathcal{A}) \, .
$$
In the particular case of Hypothesis $H$ where $\mathcal{A}=S^4$, then this already implies the corresponding hypothesis for IIA supergravity, namely that $\mathrm{Cyc}(\mathcal{A})= \mathrm{Cyc}(S^4).$ Crucially, this is obviously \textit{inequivalent} to the traditional ``Hypothesis K'' of the  IIA flux quantization in \textit{(abelian)} twisted K-theory $KU^0 \, /\!/  BU(1)$, as is apparent already at the rational level \eqref{DimReducedBianchiIdentities}, whereby the dual NS flux $H_7$ appears explicitly to contribute with its \textit{non-linear} Bianchi identity.\footnote{Let us mention, however, that there exists a different (but closely related), recently discovered admissible choice of flux quantization for 11D supergravity, whose cyclification receives a (non-rational) comparison map from a version of ``unstable K-theory'' \cite{BSS26}.} Of course, one should also note here that, in contrast to $\mathrm{Cyc}(\mathcal{A})$, the lift of abelian twisted $K$-theory to 11D is, at best, problematic. The above argument, or ``new hypothesis'', may be summarized succinctly in the following claim (cf. \cite[Claim 1.1]{GSS24-SuGra}).

\begin{claim}[{\bf Flux-quantized super-fields of 10D IIA SuGra}]
\label{FluxQuantizedSuperFieldsOf10dSugra}
For 
\begin{itemize}[leftmargin=.5cm]
\item[\bf --]
$\big(X, (e, \psi, \omega)\big)$ an $(10\vert \mathbf{16}\oplus \overline{\bf{16}})$-dimensional super-spacetime
\item[\bf --]
equipped with a superized dilaton/dilatino pair $(\Phi,\lambda)$ \eqref{DilatonDilatinoPair} such that
$$
\dd \Phi \, = \, \partial_a\Phi\, e^a \, + \, 2 \big(\,\overline{\psi} 
    \,\Gamma^{\ten}\,
    \lambda
    \big) \, ,
$$

\item[\bf --] and $\mathcal{A}$ a choice of flux quantization law as in \cite[\S 3.2]{SS24Flux} for its parent 11D supergravity embodied by a classifying 
space with rational Whitehead $L_\infty$-algebra that of the 4-sphere,
\end{itemize}
the full flux-quantized super-NS/RR-field configurations on $X$ are diagrams in
super-homotopy theory \cite[Def. 2.57]{GSS24-SuGra} of the following form:

\vspace{-6mm}
\begin{equation}
  \label{FluxQuantizedIIASugraFields}
  \hspace{-1mm} 
  \begin{tikzcd}[
    row sep=65pt, 
    column sep=50pt
  ]
    &[-20pt]
    &&
    &[+5pt]
    \overset{
      \mathclap{
      \scalebox{.65}{
        \def\arraystretch{.9}
        \begin{tabular}{c}
          \color{darkblue}
          \bf
          classifying space
          \\
          \color{darkblue}
          \bf
          of quantized
          \\
          \color{darkblue}
          \bf
          NS/D-brane charges
        \end{tabular}
      }
      \mathrlap{
        \scalebox{.7}{
          \color{gray}
          \begin{tabular}{c}
            {\rm so that}
            \\
            $\mathfrak{l}\mathrm{Cyc}(\mathcal{A}) \,\cong\,$
            \\ $ \mathrm{cyc}(\mathfrak{l}S^4)$
          \end{tabular}
        }
      }
      }    
    }{
      \mathrm{Cyc}(\mathcal{A})
    }
    \ar[
      d,
      "{
        \mathbf{ch}^{\mathrm{Cyc}(\mathcal{A})}
      }"{swap},
      "{
        \scalebox{.65}{
          \def\arraystretch{.9}
          \begin{tabular}{c}
            \color{darkgreen}
            \bf
            differential
            \\
            \color{darkgreen}
            \bf
            character
            \\
            \rm
            \cite[Def. 9.2]{FSS23Char}
          \end{tabular}
        }
      }"{xshift=-6pt}
    ]
    \\
    \underset{
      \mathclap{
        \raisebox{-2pt}{
      \scalebox{.65}{
        \def\tabcolsep{-2pt}
        \def\arraystretch{.9}
        \begin{tabular}{c}
          \color{gray}
          \bf
          ordinary
          \\
          \color{gray}
          \bf
          spacetime
        \end{tabular}
      }            
        }
      }
    }{
      \mathcolor{gray}
      {\bos{X}}
    }
    \ar[
      r,
      hook,
      gray,
      "{
        \eta^{\rightsquigarrow}_X
      }"
    ]
    &
    \;\;
    \underset{
      \mathclap{
        \raisebox{-5pt}{
      \scalebox{.65}{
        \def\tabcolsep{-2pt}
        \def\arraystretch{.9}
        \begin{tabular}{c}
          \color{darkblue}
          \bf
          super-
          \\
          \color{darkblue}
          \bf
          spacetime
        \end{tabular}
      }            
        }
      }
    }{
      X
    }
    \;\;\;
    \ar[
      rr,
      "{
        (F_2^s, H_3^s F_4^s,F_6^s, H_7^s)
      }"{name=t},
      "{
        \scalebox{.65}{
          \begin{tabular}{c}
            \color{darkgreen}
            \bf
            super-NS/RR-field flux           \eqref{CurvedNS/RRCocycles}
          \end{tabular}
        }
      }"{swap}
    ]
    \ar[
      urrr,
      bend left=17,
      "{
        \scalebox{.65}{
          \begin{tabular}{c}
            \color{darkgreen}
            \bf
            local NS/RR-field charge
          \end{tabular}
        }
      }"{sloped},
      "{
        \rchi^{IIA}
      }"{sloped, pos=.48,swap, name=s}
    ]
    \ar[
      from=s,
      to=t,
      Rightarrow,
      "{
        (\widehat{A}{}_1^s, \widehat{B}{}_2^s,
          \widehat{A}{}^s_3, \widehat{A}{}_5^s,
          \widehat{B}{}^s_6
        )
        \mathrlap{
          \scalebox{.6}{
            \color{darkorange}
            \bf
            \begin{tabular}{c}
                      global NS/RR-field        \\
                        gauge potentials
            \end{tabular}
          }
        }
      }"{description}
    ]
    &&
    \underset{
      \mathclap{
        \raisebox{-6pt}{
      \scalebox{.65}{
        \def\tabcolsep{-2pt}
        \def\arraystretch{.9}
        \begin{tabular}{c}
          \color{darkblue}
          \bf
          Smooth super-set of
          \\
          \color{darkblue}
          \bf
          duality-symmetric
          \\
          \color{darkblue}
          \bf
          NS/RR-field flux densities
        \end{tabular}
      }            
        }
      }
    }{
    \Omega^1_{\mathrm{dR}}(
      -;
      \mathrm{cyc}(\mathfrak{l}S^4)
    )_{\mathrm{clsd}}
    }
    \ar[
      r,
      "{
        \eta^{\,\scalebox{.5}{$\shape$}}
      }",
      "{
        \scalebox{.65}{
          \def\arraystretch{.9}
          \begin{tabular}{c}
            \color{darkgreen}
            \bf
            up to 
            \\
            \color{darkgreen}
            \bf
            deformations
          \end{tabular}
        }
      }"{swap}
    ]
    &
    \underset{
      \mathclap{
        \raisebox{-6pt}{
      \scalebox{.65}{
        \def\tabcolsep{-2pt}
        \def\arraystretch{.9}
        \begin{tabular}{c}
          \color{darkblue}
          \bf
          their deformation
          \\
          \color{darkblue}
          \bf
          $\infty$-groupoid
        \end{tabular}
      }            
        }
      }
    }{
    \shape
    \,
    \Omega^1_{\mathrm{dR}}(
      -;
      \mathrm{cyc}(\mathfrak{l}S^4)
    )_{\mathrm{clsd}}\,.
    }
  \end{tikzcd}
\end{equation}
The bottom part exists, by Thm. \ref{10dIIASugraEoMFromSuperFluxBianchiIdentity}, if and only if $\big(X,(e,\psi,\omega),(\Phi,\lambda)\big)$ solves the equations of motion of 10D IIA SuGra for the given flux densities $(F_2,H_3,F_4)$ and their duals .
\end{claim}

\smallskip
This proposal for the full globally completed field content of IIA supergravity dictates the global behavior (i.e., patching) of the ``(higher) gauge potentials'' for the NS/RR fluxes on all non-trivial spacetime topologies, in a manner compatible with lifting to full globally completed 11D supergravity (Cor. \ref{Compatibilitywith11DFluxQuantization}). Although 
we will not delve into more general implications of this in the current text, we do note that it immediately implies a ``re-derivation'' of the traditional (locally defined) gauge potentials (cf. \cite[Prop 1.1]{GSS24-SuGra}).
\begin{proposition}[\bf Recovering traditional super-NS/RR-field gauge potentials]
\label{RecoveringTraditionalSuperNS/RRFieldGaugePotentials}
 If the total NS/RR-field charge in diagram \eqref{FluxQuantizedIIASugraFields}
 vanishes, $[\rchi^{IIA}] = 0$ (as happens over any coordinate chart), such that the local charge equivalently factors through the point 
 \vspace{-1mm} 
 \begin{equation}
  \label{DiagramForTrivialCharge}
  \begin{tikzcd}[
    row sep=40pt, column sep=60pt
  ]
    &&[+10pt]
    \ast
    \ar[r, hook]
    &[-20pt]
    \mathrm{Cyc}(\mathcal{A})
    \ar[
      d,
      "{
        \mathbf{ch}  ^{\mathrm{Cyc}(\mathcal{A})}
      }"
    ]
    \\
    X
    \ar[
      rr,
      "{\color{darkgreen}
        (F_2^s, H^s_3, F_4^s, F^s_6, H_7^s)
      }"{name=t}
    ]
    \ar[
      urr,
      bend left=19,
      "{
        \scalebox{.7}{
          \color{darkgreen}
          \bf \small 
          \def\arraystretch{.9}
          \begin{tabular}{c}
          trivial 
          \\
          charge
          \end{tabular}
        }
      }"{sloped, pos=0.5},
      "{
        0
      }"{swap, name=s}
    ]
    \ar[
      from=s,
      to=t,
      Rightarrow,
      "{\color{orangeii} \bf 
        (A_1^s, B_2^s, A_3^s, A^s_5, B_6^s)
      }"{pos=.2}
    ]
    &&
    \Omega^1_{\mathrm{dR}}\big(-;\, \mathrm{Cyc}(\mathfrak{l}S^4)\big)_{\mathrm{clsd}}
    \ar[
      r,
      "{
        \eta^{\,\scalebox{.7}{$\shape$}}
      }"
    ]
    &
    \shape
    \,
    \Omega^1_{\mathrm{dR}}\big(-;\, \mathrm{cyc}(\mathfrak{l}S^4)\big)_{\mathrm{clsd}}
    \mathrlap{\,,}
  \end{tikzcd}
 \end{equation}
 then the NS/RR-field gauge potentials according to Claim \ref{FluxQuantizedSuperFieldsOf10dSugra}
 correspond to super-differential forms
 \begin{equation}
   \label{GlobalGaugePotentials}
   \hspace{2cm}
   \mathllap{
   \scalebox{.7}{
    \color{darkblue}
    \bf
    \begin{tabular}{c}
      ordinary 
      gauge potentials
      \\
      for the NS/RR-fields
    \end{tabular}
    }
  }
  \left.
  \def\arraystretch{1.3}
  \begin{array}{l}
    A_1^s  \,\in\,
    \Omega^1_{\mathrm{dR}}(X)
    \\
    B_2^s \,\in\,
    \Omega^2_{\mathrm{dR}}(X)\\
    A_3^s \,\in\,
    \Omega^3_{\mathrm{dR}}(X)\\
    A_5^s \,\in\,
    \Omega^5_{\mathrm{dR}}(X)\\
    B_6^s \,\in\,
    \Omega^6_{\mathrm{dR}}(X)
  \end{array}
\!\!  \right\}
  \hspace{.4cm}
  \mbox{\rm \small such that}
  \hspace{.4cm}
  \left\{\!\!
  \def\arraystretch{1.2}
  \begin{array}{l}
    \mathrm{d}
    \, 
    A_1^s 
      \;=\; 
    F_2^s
    \,,
    \\
    \mathrm{d}\, 
    B_2^s 
      \;=\;
   H_3^s
    \,,
    \\
    \mathrm{d}\, 
    A_3^s 
      \;=\;
    F_4^s 
    -
    F_2^s\, B_2^s
    \,,
     \\
    \mathrm{d}\, 
    A_5^s 
      \;=\;
    F_6^s 
    +\tfrac{1}{2}A_3^s \, H_3^s + \tfrac{1}{2} B_2^s \,F_4^s
    \,,
     \\
    \mathrm{d}\, 
     B_6^s 
      \;=\;
    H_7^s 
    -
    \tfrac{1}{2} A_3^s \, F_4^s - F_2^s \, A_5^s
    \,,
  \end{array}
  \right.
 \end{equation}
 as traditionally considered in the IIA literature (see, e.g., \cite[(A.19)]{BNS04}).\footnote{Up to terms due to ``integration by parts'', which are dropped when studied in the standard IIA Lagrangian.}
\end{proposition}
\begin{proof}
  By \cite[Ex. 2.55]{GSS24-SuGra}, homotopies in \eqref{DiagramForTrivialCharge} correspond to null-coboundaries for $\big(F^s_2, H_3^s, F_4^s, F_6^s, H_7^s \big)$ $\in\, \Omega^1_{\mathrm{dR}}\big(X;\, \mathrm{cyc}(\mathfrak{l}S^4)\big)_\mathrm{clsd}$, and by Prop. \ref{CycS4CoboundariesAreRRpotentials} below  these correspond to exactly the gauge potentials \eqref{GlobalGaugePotentials}.
\end{proof}




\section{
\texorpdfstring{KK reduction of flux fields via cyclification of classifying $L_\infty$-algebras}
{KK reduction of flux fields via cyclification of classifying L-infinity-algebras}
}

We proceed immediately to globalize the tangent-space wise $\mathbb{R}$-algebraic (double) dimensional reduction via the cyclification/oxidation adjunction of $L_\infty$-algebras (reviewed in \cite[\S 2.2]{BMSS19}\cite[\S 2.2]{GSS24-TDuality}), viewed as fibered over points, to the case of (non-trivial) tangent $L_\infty$-algebroids over supermanifolds.

\begin{lemma}[\bf Cyclification/Oxidation on supermanifolds]\label{CyclificationOxidationOnSuperManifolds} Let $X$ be a supermanifold and $\mathfrak{h}$ a super-$L_\infty$-algebra. Fix a choice of normalized and left-invariant fiber-wise coframe along the $S^1$-fiber on the product $X\times S^1$ (equiv. connection), i.e.,
$$
e' \, = \, \dd \phi + A(x,\phi)  \;\; \in \; \Omega^1_\mathrm{dR}(X\times S^1) \, ,
$$
where $\phi$ is the angular coordinate on $S^1$. Then there is a canonically associated bijection between:
\vspace{1mm} 
\begin{itemize} 
\item[\bf (a)]  $S^1$-invariant maps of super-$L_\infty$-algebroids\footnote{ The diagrams in braces could be extended to indicate which manifolds these algebroids are fibered, i.e., over $X\times S^1$ and the point $*$, respectively. We suppress these to keep the presentation as light as possible.} out of the tangent Lie algebroid $T(X\times S^1)$ into $\mathfrak{h}$ (as an $L_\infty$-algebroid over the point $*$),  
  
\item[\bf (b)] maps of super-$L_\infty$-algebroids out of $TX$ into $\mathrm{cyc}(\mathfrak{h})$ that preserve the curvature 2-form $F_2:= \dd e'$,
  \begin{equation}
    \label{LinftyAlgebroidCyclificationHomIsomorphism}
    \Big\{\!\!
    \begin{tikzcd}
      T(X\times S^1)
      \ar[
        rr,
        "{ F }"
      ]
      &&
      \mathfrak{h}
    \end{tikzcd}
   \! \!\Big\}
    \begin{tikzcd}[
      column sep=85pt
    ]
      \ar[
        r,
        shift left=4pt,
        "{  \scalebox{.7}{\color{darkgreen}
            \bf
            reduction}\;\;
          \mathrm{rdc}_{F_2}
        }",
        "{ \sim }"{swap, yshift=-2pt}
      ]
      \ar[
        r,
        <-,
        shift right=4pt,
        "{ \scalebox{.7}{
            \color{darkgreen}
            \bf
            oxidation
          }\;\;
          \mathrm{oxd}_{F_2}
        }"{swap},
      ]
      &
      {}
    \end{tikzcd}
    \bigg\{\!\!
    \begin{tikzcd}[row sep=-3pt, column sep=large]
      TX
      \ar[
        rr,
        "{ \widetilde F }"
      ]
      \ar[
        dr,
        "{ F_2 }"{swap}
      ]
      &&
      \mathrm{cyc}(\mathfrak{h})
      \ar[
        dl,
        "{ \omega_2 }"
      ]
      \\
      &
      b \mathbb{R}
    \end{tikzcd}
    \!\!\!\bigg\}
  \end{equation}
  given by
\begin{equation}
  \label{LinftyAlgebroidCyclificationHomBijection}
  \hspace{-2.7cm} 
  \begin{tikzcd}[sep=0pt]
    T(X\times S^1)
    \ar[
      rr,
      "{ F }"
    ]
    &&
    \mathfrak{h}
    \\
   F^i \, = \,  F^i_{\mathrm{nw}}
    +
    e' \, \int_{S^1}  F^i
    &\longmapsfrom&
    e^i
  \end{tikzcd}
  \hspace{1.3cm}
  \leftrightsquigarrow
  \hspace{1.7cm}
  \begin{tikzcd}[
    row sep=-4pt, 
    column sep=0pt]
    TX
    \ar[
      rr,
      "{ 
        \widetilde{F} 
      }"
    ]
    &&
    \mathrm{cyc}(\mathfrak{h})
    \\
    F^i_{\mathrm{nw}}
    &\longmapsfrom&
    e^i
    \\
    -
    \int_{S^1} F^i
    &\longmapsfrom&
    \mathrm{s}e^i
    \\
    F_2
    &\longmapsfrom&
    \omega_2
    \mathrlap{\,,}
  \end{tikzcd}
\end{equation}
where on the left is the unique decomposition of any $F^i$ such that $\iota_{\xi}F^i_{\mathrm{nw}}=\int_{S^1} F^i_{\mathrm{nw}}=0$. 

\end{itemize} 
\end{lemma}
\begin{proof}
Morphisms of $L_\infty$-algebroids out of $T(X\times S^1)$ into $\mathfrak{h}$ are equivalently morphisms of sCDGAs $\mathrm{CE} (\mathfrak{h}) \rightarrow \Omega^\bullet_{\mathrm{dR}}(X\times S^1)$, hence in particular completely determined by the image of generators of $\mathrm{CE}(\mathfrak{h})$
$$
e^i \longmapsto  F^i \, = \,  F^i_{\mathrm{nw}}
    +
    e' \, \int_{S^1}  F^i ,
$$
where we decomposed the forms $F^i$ into winding and non-winding modes via the choice of coframe $e'$ along the $S^1$-factor. The condition of $S^1$-invariance on such morphisms corresponds to the vanishing
of the Lie derivatives
$$
L_{\xi} F^i \, \equiv \, \dd \iota_{\xi} F^i+ \iota_{\xi} \dd F^i\,= \,  0\, , 
$$
along the (vertical) vector field $\xi = \frac{\partial}{\partial \phi}$ generating the $S^1$-action. By the assumed $S^1$-invariance of $e'$ and the fact that 
$$
\iota_{\xi} F^i_{\mathrm{nw}}\, =0 \, =\,  \iota_{\xi} \int_{S^1} F^i ,
$$
the $S^1$-invariance $L_\xi F^i=0$ is equivalent to 
$$
\iota_{\xi} \dd F^i_{\mathrm{nw}} \, =0\, = \, \iota_{\xi} \dd \int_{S^1} F^i .  
$$
This means that the winding and non-winding forms in the decomposition of an $S^1$-invariant form $F^i$ are \textit{basic} with respect to the projection $X\times S^1\rightarrow X$, hence arise as pullbacks of some unique forms on the base $X$. This yields the bijection of morphisms at the level of super-commutative graded algebras.

At this stage, the fact that $F^*$ preserves the differentials on the left if and only if $\tilde{F}^*$ preserves the corresponding differentials on the right, or equivalently that the fluxes $\{F^{i}\}_{i\in I}$ satisfy the  $\mathfrak{h}$-Bianchi identities if and only if $\big\{F^i_{\mathrm{nw}},\, \int_{S^1} F^i\big\}_{i\in I}$ satisfy the $\mathrm{cyc}(\mathfrak{h})$-Bianchi identities, follows the exact same formal steps as the $L_\infty$-algebraic proof from \cite[Prop. 2.25]{GSS24-TDuality} (i.e., at the tangent-space level).
\end{proof}

The content of the Lemma can be restated as the association to any choice of connection on $X\times S^1$ of a bijection of $S^1$-invariant, closed $\mathfrak{h}$-valued forms on $X\times S^1$ and closed $\mathrm{cyc}(\mathfrak{h})$-valued forms on $X$ with fixed $2$-form component $F_2 := \dd e' $
$$
\Omega^1_{\mathrm{dR}}\big(X\times S^1; \,\mathfrak{h}\big)_{\mathrm{ \mathrm{clsd}, \, inv.}} \; \cong_{e'} \; \Omega^1_{\mathrm{dR}}\big(X; \,\mathrm{cyc}(\mathfrak{h})\big)_{ \mathrm{clsd}, \, F_2} \, .
$$
 
For exposition, we have stated the result for the trivial circle bundle $X\times S^1$, hence equivalently with trivial 1st Chern class $[F_2]=0\in H_{\mathrm{dR}}^2(X)$. However, the proof holds verbatim for any non-trivial (super) principal $S^1$-bundle 
classified by a non-trivial Chern class
$$
[F_2] \;\; \in \;\; H_{\mathrm{dR}}^2(X;\mathbb{Z})\, ,
$$
i.e., with closed (and integral) curvature $\dd F_2 =0$ but not (globally) exact on X. In fact, this further extends to an isomorphism between the corresponding  integrating $\infty$-groupoids of (higher) concordances ($\equiv$ gauge potentials, cf. Cor. \ref{CycS4CoboundariesAreRRpotentials}).

\begin{proposition}[\bf Supermanifold cyclification/oxidation for non-trivial circle bundles] \label{NonTrivialCircleBundleCyclification/Oxidation}
Let $X$ be a supermanifold supplied with an even, integral closed but potentially non-exact 2-form $ F_2\in \Omega^2_{\mathrm{dR}}(X;\mathbb{Z})$ and $\mathfrak{h}$ a super-$L_\infty$-algebra. For any representative $S^1$-principal bundle 
$$
S^1\longhookrightarrow Y \xlongrightarrow{\pi} X\, ,
$$ 
classified by the Chern class $[F_2]\in H^2_{\mathrm{dR}}(X, \mathbb{Z})$, fix a choice of connection\footnote{We abuse the nomenclature as usual in the literature. The actual $U(1)$-connections and curvatures are, of course, valued in the Lie algebra $\mathrm{Lie}(U(1))\cong i \FR$, hence related to our forms as $i e' \in \Omega^1_\mathrm{dR}(Y; i \FR)$ and $i F_2 \in \Omega^2_\mathrm{dR}(X;i\FR) $.}
$e' \in \Omega^1_{\mathrm{dR}}(Y)$
trivializing the (pullback) curvature
$$
\dd e' \, = \, \pi^* F_2\, ,
$$
hence in particular so that 
\vspace{1mm} 
$$
 \iota_\xi e'\equiv \int_{S^1}e'  \, = \, 1\,,
$$
for $\xi$ the fundamental (vertical) vector field generating the $S^1$-action on $Y$. Then:
\vspace{1mm} 
\begin{itemize} 
\item[{\bf (i)}] There is a canonically associated bijection between:
\vspace{1mm} 
\begin{itemize} 
\item[\bf (a)] $S^1$-invariant maps of super-$L_\infty$-algebroids out of the tangent Lie algebroid $TY$ into $\mathfrak{h}$,  
  
\item[\bf (b)] maps of super-$L_\infty$-algebroids out of $TX\cong T(Y/S^1)$ into $\mathrm{cyc}(\mathfrak{h})$ that preserve the curvature 2-form $F_2$,
  \begin{equation}
    \label{NontrivialLinftyAlgebroidCyclificationHomIsomorphism}
    \Big\{\!\!
    \begin{tikzcd}
      TY
      \ar[
        rr,
        "{ F }"
      ]
      &&
      \mathfrak{h}
    \end{tikzcd}
   \! \!\Big\}
    \begin{tikzcd}[
      column sep=85pt
    ]
      \ar[
        r,
        shift left=5pt,
        "{  \scalebox{.7}{\color{darkgreen}
            \bf
            reduction}\;\;
          \mathrm{rdc}_{F_2}
        }",
        "{ \sim }"{swap, yshift=-2pt}
      ]
      \ar[
        r,
        <-,
        shift right=5pt,
        "{ \scalebox{.7}{
            \color{darkgreen}
            \bf
            oxidation
          }\;\;
          \mathrm{oxd}_{F_2}
        }"{swap},
      ]
      &
      {}
    \end{tikzcd}
    \bigg\{\!\!\!
    \begin{tikzcd}[row sep=-3pt, column sep=large]
      TX
      \ar[
        rr,
        "{ \widetilde F }"
      ]
      \ar[
        dr,
        "{ F_2 }"{swap}
      ]
      &&
      \mathrm{cyc}(\mathfrak{h})
      \ar[
        dl,
        "{ \omega_2 }"
      ]
      \\
      &
      b \mathbb{R}
    \end{tikzcd}
    \!\!\!\bigg\}
  \end{equation}
  given by
\begin{equation}
  \label{NontrivialLinftyAlgebroidCyclificationHomBijection}
  \hspace{-2.7cm} 
  \begin{tikzcd}[sep=0pt]
    TY
    \ar[
      rr,
      "{ F }"
    ]
    &&
    \mathfrak{h}
    \\
   F^i \, = \,  F^i_{\mathrm{nw}}
    +
    e' \, \int_{S^1}  F^i
    &\longmapsfrom&
    e^i
  \end{tikzcd}
  \hspace{1cm}
  \leftrightsquigarrow
  \hspace{1cm}
  \begin{tikzcd}[
    row sep=-4pt, 
    column sep=0pt]
    TX
    \ar[
      rr,
      "{ 
        \widetilde{F} 
      }"
    ]
    &&
    \mathrm{cyc}(\mathfrak{h})
    \\
    F^i_{\mathrm{nw}}
    &\longmapsfrom&
    e^i
    \\
    -
    \int_{S^1} F^i
    &\longmapsfrom&
    \mathrm{s}e^i
    \\
    F_2
    &\longmapsfrom&
    \omega_2
    \mathrlap{\,.}
  \end{tikzcd}
\end{equation}
That is, associated to any choice of $U(1)$-connection $e'$ on $Y$, there is a bijection of $S^1$-invariant, closed $\mathfrak{h}$-valued forms on $Y$ and closed $\mathrm{cyc}(\mathfrak{h})$-valued forms on $X\cong Y/S^1$ with fixed 2-form component $F_2 := \dd e'$
$$
\Omega^1_{\mathrm{dR}}\big(Y; \,\mathfrak{h}\big)_{\mathrm{clsd},\, \mathrm{inv.}} \;\; \cong_{e'} \;\; \Omega^1_{\mathrm{dR}}\big(X; \,\mathrm{cyc}(\mathfrak{h})\big)_{ \mathrm{clsd}, \, F_2}\, .
$$
\end{itemize} 

\item[{\bf (ii)}]
This extends to a (strict) isomorphism of the corresponding $\infty$-groupoids of fluxes with their higher (globally defined)\footnote{In particular, depending on the topology of $X$ these might not exist. Hence, strictly speaking, this result applies non-trivially patch-wise on the $
\infty$-groupoids of the restricted fluxes on charts $U\hookrightarrow X$ (cf. Prop. \ref{CycS4CoboundariesAreRRpotentials}). } concordances 
$$
\int  \Omega^1_{\mathrm{dR}}\big(Y; \,\mathfrak{h}\big)_{\mathrm{ \mathrm{clsd}, inv.}} \, \cong_{e'}  
\int  \Omega^1_{\mathrm{dR}}\big(X; \,\mathrm{cyc}(\mathfrak{h})\big)_{ \mathrm{clsd}, F_2} \;\;
\xhookrightarrow{\quad} \; \int  \Omega^1_{\mathrm{dR}}\big(X; \,\mathrm{cyc}(\mathfrak{h})\big)_{ \mathrm{clsd}}  \quad \in \quad \mathrm{SimpSet}_{\mathrm{Kan}} \, .
$$
\end{itemize}
\end{proposition}
\begin{proof}
The proof of Lemma \ref{CyclificationOxidationOnSuperManifolds} continues to hold, with the only caveat that the coordinate form of the connection $e'=d\phi + A(x,\phi)$ holds only locally over $X$, i.e., on local trivializations $Y|_{U} \cong U\times S^1$. 

The extension to the $\infty$-groupoids of higher concordances follows by lifting the bijection \eqref{NontrivialLinftyAlgebroidCyclificationHomIsomorphism} to have instead as domains the $L_\infty$-algebroids $T(Y\times [0,1]^{\times n})$ and $T(X\times [0,1]^{\times n})$ that support (higher) concordances, using the induced pulled back U(1) connection (denoted by the same symbol $e'$). Explicitly, for instance for the case of $n=1$, we have\footnote{For $n\geq 2$, it might seem that this formulation would land us in the world of \textit{cubical sets} rather than a simplicial sets. However, since by definition higher concordances are constant along the original form data (cf. \cite[Def. 2.46]{GSS24-SuGra}) this resulting cubical set is, roughly speaking, identified with a corresponding simplicial set.} 
\vspace{-2mm} 
\begin{equation*}
  \hspace{-1cm} 
  \begin{tikzcd}[sep=0pt]
    T\big( Y \times [0,1]_t\big )
    \ar[
      rr,
      "{ \widehat{F}= F+ \dd t\, A \;}"
    ]
    &&
    \mathfrak{h}
    \\
   \widehat{F}^i \, = \,  \widehat{F}^i_{\mathrm{nw}}
    +
    e' \, \int_{S^1}  \widehat{F}^i
    &\longmapsfrom&
    e^i
  \end{tikzcd}
  \hspace{1cm}
  \leftrightsquigarrow
  \hspace{1cm}
  \begin{tikzcd}[
    row sep=1pt, 
    column sep=0pt]
    T\big(X\times [0,1]_t\big) 
    \ar[
      rr,
      "{ 
        \widetilde{\widehat{F}} 
      }"
    ]
    &&
    \mathrm{cyc}(\mathfrak{h})
    \\
    \widehat{F}^i_{\mathrm{nw}}= F^i_\mathrm{nw} + \dd t\, A_{\mathrm{nw}}^{i} 
    &\longmapsfrom&
    e^i
    \\
    -
    \int_{S^1} \widehat{F}^i = -
    \int_{S^1} F^i - \dd t \, \int_{S^1} A^i 
    &\longmapsfrom&
    \mathrm{s}e^i
    \\
    \widehat{F}_2= t\,  F_2 + \dd t \, e'
    &\longmapsfrom&
    \omega_2
    \mathrlap{\,.}
  \end{tikzcd}
\end{equation*}

 \vspace{-5mm} 
\end{proof}

\begin{remark}[\bf Effect of choice of connection]\label{EffectOfConnectionChoice}
$\,$
\begin{itemize} 
\item[\bf (i)] A gauge equivalent choice of $U(1)$-connection 
$$
\overline{e} = e' - i \,u \,  \dd u^{-1} = e' + \dd N_0 
$$ 
on $Y$, where $u = e^{i N_0}\, :\,  Y \longrightarrow U(1)$ is an equivariant map ($\equiv$ vertical automorphism of $Y$), for some $N_0 \in \Omega^0(Y)$, yields the same curvature $F_2$ but a different decomposition of invariant forms as
\begin{align*}
 F^i \, &= \,  \overline{F}^i_{\mathrm{nw}}
    +
   \overline{e} \, \int_{S^1}  F^i \, =\, \overline{F}^i_{\mathrm{nw}} + (e' + \dd N_0 ) \, \int_{S^1}  F^i \, \\
&=\, \bigg( \overline{F}^i_{\mathrm{nw}}  +\dd N_0  \wedge \, \int_{S^1}  F^i \bigg)  + e'\, \int_{S^1}  F^i , 
\end{align*} 
so that 
\vspace{2mm} 
$$
\overline{F}^i_{\mathrm{nw}} = F^i_{\mathrm{nw}} -\dd N_0 \, \wedge \int_{S^1}  F^i  . 
$$
Hence, the formula for the corresponding bijection of $L_\infty$-algebroid morphisms from \eqref{NontrivialLinftyAlgebroidCyclificationHomBijection} differs accordingly.

\item[\bf (ii)] A potentially gauge inequivalent choice of connection $\widetilde{e}= e' + a$, where $a \in \Omega^1_\mathrm{dR}(Y)$ is an arbitrary horizontal and equivariant $1$-form, yields similarly a different decomposition with
$$
\widetilde{F}^i_{\mathrm{nw}} = F^i_{\mathrm{nw}} - a \wedge \int_{S^1} F^i  ,
$$
but now starting with  a different curvature $2$-form $\widetilde{F}_2 = F_2 + \dd a$ on $X$, further (minimally) modifying  the bijection from \eqref{NontrivialLinftyAlgebroidCyclificationHomBijection}.
\end{itemize} 
\end{remark}

In the field theoretic context, and in particular 11D supergravity \cite[Prop. 2.48]{GSS24-SuGra} with closed $\mathfrak{l}S^4$-valued flux forms, higher concordances correspond to \text{local} (higher) gauge potentials / coboundaries between gauge-equivalent fluxes. Thus, part {\bf (ii)} of the proposition may be interpreted as a rigorous and flux quantization consistent formulation of the (algorithmic) fact that $S^1$-invariant (higher) gauge potentials for the $C$-field of 11D supergravity may be identified with gauge potentials for the dimensionally reduced IIA NS/RR-field for 10D supergravity.

\begin{proposition}[\bf Coboundaries for closed cyc$(\mathfrak{l}S^4)$
-valued forms are local NS/RR potentials]\label{CycS4CoboundariesAreRRpotentials} 

\noindent Given a $(F_2, H_3, F_4, F_6, H_7) \, \in \Omega^{1}_{\mathrm{dR}}\big (X; \, \mathrm{cyc}(\mathfrak{l}S^4)\big)_{\mathrm{clsd}}$ and any local chart $\iota_U : U \hookrightarrow X$,

\smallskip 
\begin{itemize}[itemsep=1pt,leftmargin=.75cm]
\item[{\bf (i)}] there is a natural surjection 

\begin{itemize}[leftmargin=.9cm]
\item
from local null-coboundaries \cite[Def. 3.91]{FSS23Char} for $(F_2, H_3, F_4, F_6, H_7)$ in $\mathrm{cyc}(\mathfrak{l}S^4)$-valued de Rham cohomology, 
\item
to 5-tuples of ordinary differential
forms
\begin{equation}
  \label{TuplesOfPotentialForms}
  \hspace{-1cm} 
  \left.
  \def\arraystretch{1.3}
  \begin{array}{l}
    A_1  \,\in\,
    \Omega^1_{\mathrm{dR}}(U)
    \\
    B_2 \,\in\,
    \Omega^2_{\mathrm{dR}}(U)\\
    A_3 \,\in\,
    \Omega^3_{\mathrm{dR}}(U)\\
    A_5 \,\in\,
    \Omega^5_{\mathrm{dR}}(U)\\
    B_6 \,\in\,
    \Omega^6_{\mathrm{dR}}(U)
  \end{array}
\!\!  \right\}
  \hspace{.4cm}
  \mbox{\rm \small such that}
  \hspace{.4cm}
  \left\{\!\!
  \def\arraystretch{1.2}
  \begin{array}{l}
    \mathrm{d}
    \, 
    A_1 
      \;=\; 
    F_2
    \,,
    \\
    \mathrm{d}\, 
    B_2 
      \;=\;
   H_3
    \,,
    \\
    \mathrm{d}\, 
    A_3 
      \;=\;
    F_4 
    -
    F_2\, B_2
    \,,
     \\
    \mathrm{d}\, 
    A_5 
      \;=\;
    F_6 
    +\tfrac{1}{2}A_3 \, H_3 + \tfrac{1}{2} B_2 \,F_4
    \,,
     \\
    \mathrm{d}\, 
     B_6 
      \;=\;
    H_7 
    -
    \tfrac{1}{2} A_3 \, F_4 - F_2 \, A_5
    \,.
  \end{array}
  \right.
\end{equation}
\end{itemize}

\item[{\bf (ii)}] This surjection respects local equivalence classes, where
\begin{itemize}[leftmargin=.9cm]
\item equivalence of $\mathrm{cyc}(\mathfrak{l}S^4)$-coboundaries is by coboundaries of coboundaries \cite[Def. 2.46]{GSS24-SuGra},
\item (finite gauge) equivalence of the tuples \eqref{TuplesOfPotentialForms} is defined as follows:  
\begin{align}
  \label{EquivalenceOfTraditionalRRFieldGaugePotentials}
  (A_1,\,B_2,\, A_3,\, A_5,\,  B_6)
  \;\sim\; &
    (A'_1,\,B'_2,\, A'_3,\, A'_5,\,  B'_6)
  \hspace{.5cm}
  \nonumber 
  \\
  \Updownarrow
  \nonumber 
  \\
  \exists
  \,
  \left.
  \def\arraystretch{1.2}
  \begin{array}{l}
  N_0 \,\in\,\Omega^0_{\mathrm{dR}}(U)
  \\
  N_1 \,\in\,\Omega^1_{\mathrm{dR}}(U)
   \\
  M_2 \,\in\,\Omega^2_{\mathrm{dR}}(U)
   \\
  N_4 \,\in\,\Omega^4_{\mathrm{dR}}(U)
   \\
  M_5 \,\in\,\Omega^5_{\mathrm{dR}}(U)
  \end{array}
  \!\!\! \right\}
  \;\;
  \mbox{\rm \small such that} &
  \;\;
  \left\{\!\!\!
  \def\arraystretch{1.1}
  \begin{array}{l}
    \mathrm{d}\, N_0 \;=\;
    A_1' - A_1\, , 
    \\
      \mathrm{d}\, N_1 \;=\;
    B'_2 - B_2 \, ,
    \\
    \mathrm{d}\, M_2 
      \;=\; 
    A'_3 - A_3
    + F_2\, N_1 
    \,,
    \\
      \mathrm{d}\, N_4 \;=\; A'_5 - A_5 - \frac{1}{2} (B'_2 \, A_3 - A'_3 \, B_2)\,, 
     \\
      \mathrm{d}\, M_5 \;=\;
     B'_6 - B_6 + F_2 \,  N_4 - \frac{1}{2} A'_3 \,  A_3 \,  . 
  \end{array}
  \right.
\end{align}
\end{itemize}
\item[{\bf (iii)}] In the case where $F_2$ is already taken to be flux-quantized in $\mathrm{B}U(1)\cong \mathrm{B}^2 \mathbb{Z}$, hence with integral $[F_2] \in H^2_{\mathrm{dR}}(X;\mathbb{Z})$, the local gauge potentials $A_1$ arise as pullbacks of connections $e'$ on the corresponding $S^1$-bundle $Y$ over $X$. Accordingly, their gauge transformations as vertical automorphisms of the bundle $Y$.
\end{itemize}
\end{proposition}

\begin{proof}
Parts \textbf{(i)} and \textbf{(ii)} follow by essentially redoing the $C$-field computation of the proof of \cite[Prop. 2.48]{GSS24-SuGra} from scratch, but with the target classifying $L_\infty$-algebra being instead $\mathrm{cyc}(\mathfrak{l}S^4)$ -- hence without even assuming $F_2$ is integral and further fixing its connection. Alternatively by assuming that $F_2$ is integral, which is indeed the case for our application to IIA supergravity,
this follows by dimensionally reducing the result of \cite[Prop. 2.48]{GSS24-SuGra}.

Indeed, let $(F_2, H_3, F_4, F_6, H_7)$ be a cyc$(\mathfrak{l}S^4)$-cocycle on $X$ with integral $[F_2]\in H^2_{\mathrm{dR}}(X,\mathbb{Z})$. Fix a connection $e'$ on the corresponding principal $S^1$-bundle classified by $[F_2]$ with $\dd e' = F_2 $. By the corresponding isomorphism of $\infty$-groupoids from Prop. \ref{NonTrivialCircleBundleCyclification/Oxidation}, a local null-coboundary $\widehat{F}: T(U\times [0,1]_t) \rightarrow \mathrm{cyc}(\mathfrak{l}S^4)$ for $(F_2= \dd e', H_3, F_4, F_6, H_7)$ is equivalently a null-coboundary $\widehat{G} : T(Y\times [0,1]_t) \rightarrow \mathfrak{l}S^4$ for 
\begin{equation}\label{DecomposedG4G7}
\big(G_4 := F_4 - e' H_3\, , \, G_7 := H_7 - e' F_6\big) \quad \in \quad \Omega^{1}_\mathrm{dR}(Y;\, \mathfrak{l}S^4)_{\mathrm{clsd}} \, ,
\end{equation}
since the coboundary for $\dd F_2 = 0 $ has already been fixed via the choice of $A_1= e'$. \footnote{Here, we abuse notation and denote the (locally defined) gauge potential on $U\hookrightarrow X$ by the same symbol as the corresponding globally defined connection on $Y$.} By \cite[Prop. 2.48]{GSS24-SuGra}, to such a null coboundary corresponds 
(an equivalence class of) pair(s) of ordinary (gauge potential) differential
forms
\begin{equation}
  \label{PairsOfPotentialForms}
  \left.
  \def\arraystretch{1.3}
  \begin{array}{l}
    C_3 \,\in\,
    \Omega^3_{\mathrm{dR}}(Y_U)
    \\
    C_6 \,\in\,
    \Omega^6_{\mathrm{dR}}(Y_U)
  \end{array}
\!\!  \right\}
  \hspace{.4cm}
  \mbox{\rm \small such that}
  \hspace{.4cm}
  \left\{\!\!
  \def\arraystretch{1.2}
  \begin{array}{l}
    \mathrm{d}
    \, 
    C_3 
      \;=\; 
    G_4
    \,,
    \\
    \mathrm{d}\, 
    C_6 
      \;=\;
    G_7 
    -
    \tfrac{1}{2} C_3 \, G_4
    \,.
  \end{array}
  \right.
\end{equation}
Decomposing these in terms of basic forms, uniquely with respect to $e'$, as
\begin{align}\label{CpotentialDecomposition}
C_3 \, = \, A_3 + e'\,  B_2  \, , \qquad C_6 \, = \, B_6 + e' \, A_5  
\end{align}
and substituting in the differential condition \eqref{PairsOfPotentialForms}, together with the decomposition of the $(G_4, G_7)$ from \eqref{DecomposedG4G7}, results precisely in a 5-tuple $(A_1= e' , B_2, A_3, A_5, B_6)$ satisfying the differential conditions \eqref{TuplesOfPotentialForms}. For instance, the first relation from \eqref{PairsOfPotentialForms} decomposes in terms of basic forms as
\begin{align*}
\dd A_3 + F_2\, B_2 - e' \, \dd B_2 = F_4 - e'\, H_3   
\end{align*}
which is equivalent to $\dd A_3 = F_4  - F_2 \, B_2 $ and $\dd B_2 = H_3 $, with the remaining two relations following by analogously decomposing $\dd C_6 = G_7 - \tfrac{1}{2}C_3 \, G_4$ in terms of basic forms. 

Next, consider a concordance $\widehat{\widehat{F}}: T\big(U\times [0,1]_t \times [0,1]_s \big)$ between two (local) null-coboundaries $\widehat{F}, \widehat{F}'$ for the original cocycle $(F_2= \dd e', H_3, F_4, F_6, H_7)$. Then, again by the isomorphism of $\infty$-groupoids corresponding to the choice of connection $e'$ from Prop. \ref{NonTrivialCircleBundleCyclification/Oxidation}, these correspond precisely to a concordance 
$$
\widehat{\widehat{G}} : T\big(Y_U\times [0,1]_t\times [0,1]_s\big) \longrightarrow \mathfrak{l}S^4
$$
between the lifted coboundaries $\widehat{G},\widehat{G}'$ corresponding to $\widehat{F}$ and $\widehat{F}'$, respectively. By \cite[Prop. 2.48]{GSS24-SuGra}, to such a concordance corresponds a pair of (gauge-of-gauge potential, or equivalently, gauge transformations) differential forms  \begin{equation}
  \label{PairsOfGaugeOfGaugePotentialForms}
  \left.
  \def\arraystretch{1.3}
  \begin{array}{l}
    L_2 \,\in\,
    \Omega^2_{\mathrm{dR}}(Y_U)
    \\
    L_5 \,\in\,
    \Omega^5_{\mathrm{dR}}(Y_U)
  \end{array}
\!\!  \right\}
  \hspace{.4cm}
  \mbox{\rm \small such that}
  \hspace{.4cm}
  \left\{\!\!
  \def\arraystretch{1.2}
  \begin{array}{l}
    \mathrm{d}
    \, 
     L_2 
      \;=\; 
    C_3' - C_3
    \,,
    \\
    \mathrm{d}\, 
    L_5 
      \;=\;
    C_6' - C_6 
    -
    \tfrac{1}{2} C_3' \, C_3
    \,.
  \end{array}
  \right.
\end{equation}
Decomposing these in terms of basic forms, uniquely with respect to $e'$, as
$$
L_2 \, = \, M_2 - e'\,  N_1 \, , \qquad L_5 \, = \, M_5 - e' \, N_4  
$$
and substituting in the differential relation \eqref{PairsOfGaugeOfGaugePotentialForms}, along with the decomposition of the correspond 11D C-field gauge potentials from \eqref{CpotentialDecomposition}, results into the following equivalent basic conditions
\begin{align*}
    \mathrm{d}\, N_1 \;&=\;
    B'_2 - B_2 \, ,
    \\
\mathrm{d}\, M_2 
      \;&=\; 
    A'_3 - A_3
    + F_2\, N_1
    \,,
    \\
 \mathrm{d}\, N_4 \;&=\; A'_5 - A_5 - \tfrac{1}{2} (B'_2 \, A_3 - A'_3 \, B_2)\,, 
     \\
\mathrm{d}\, M_5 \;&=\;
     B'_6 - B_6 + F_2 \,  N_4 - \tfrac{1}{2} A'_3 \,  A_3 \, .
\end{align*}

Unsurprisingly, these have the form of ``(finite) gauge transformations'' for the IIA gauge potentials, but with a vanishing transformation on the gauge potential $A_1$ for the $2$-flux $F_2$. Of course, this is to be expected since in deriving these, we passed through the isomorphism of Prop. \ref{NonTrivialCircleBundleCyclification/Oxidation} corresponding to a \textit{fixed} connection $A_1=e'$ for $F_2$ (i.e., global flux-quantized potential). To recover the remaining transformation, one simply repeats the above algorithm, but by choosing any gauge-equivalent connection $A_1'=e' + \dd N_0 $ for $N_0 \in C^\infty(Y)$, i.e.,
$$\dd N_0 \, = \, A_1' - A_1$$
and passing through the corresponding isomorphism of $\infty$-groupoids, which essentially amounts to a correspondingly
modified decomposition of the lifted 11D differential forms (Rem. \ref{EffectOfConnectionChoice}). The latter is easily seen not to propagate the transformation $N_0$ to the rest of the (higher) gauge fields. 
\end{proof}

\begin{remark}[\bf The Lie algebroid of $S^1$-invariant forms]\label{LieAlgebroidOfS1InvariantForms}
As an aside of mathematical interest, we note that the $S^1$-invariance condition of differential forms on the total space $Y$ may also be recast in the $L_\infty$-algebroid language, by considering instead the source to be a quotient Lie algebroid of $TY$. Indeed, by the fully faithful inclusion of $L_\infty$-algebroids into the opposite category of sCDGAs, the canonical injection of $S^1$-invariant forms 
$$
\Omega^{\bullet}_{\mathrm{dR}}(Y)_{\mathrm{inv.}} \longhookrightarrow \Omega^{\bullet}_{\mathrm{dR}}(Y)
$$
must correspond dually to a projection of Lie algebroids
$$
TY \longrightarrow  E\, . 
$$
The choice of a $U(1)$-connection allows for an explicit identification of the Lie algebroid structure on $E$ via the induced isomorphism
$$
\Omega^\bullet_{\mathrm{dR}}(X)[e']\, \cong \,  \Omega^\bullet_{\mathrm{dR}}(X)\oplus \Omega^\bullet_{\mathrm{dR}}(X)\cdot e' \xrightarrow{\quad \sim \quad } \Omega^{\bullet}_{\mathrm{dR}}(Y)_{\mathrm{inv.}} \,.
$$
Upon identifying the new generator $e'$ with the unit section  
$$
e'\equiv 1_X \quad \in \quad \Gamma_X(\mathbb{R}^*_X)
$$
of the dual to the trivial line bundle $\mathbb{R}_X = X\times \FR \rightarrow X$, the left hand side is further identified as the exterior algebra of sections
$$
\bigwedge^\bullet_{C^\infty(X)} \Gamma_X(T^*X \oplus \mathbb{R}^*_X)\, .
$$
This description makes it manifest that the sCDGA of invariant forms on a principal $S^1$-bundle (supplied with a  connection $e'$ of curvature $F_2$) is (isomorphic to) the Chevalley--Eilenberg algebra of the Lie algebroid
$$E\, \cong \, TX\oplus \mathbb{R}_X$$
with anchor map the projection
$\pi_1 \, : \, TX \oplus \mathbb{R}_X \longrightarrow TX$
and bracket
$$
[X+f,\, Y+g] \, := 
\,  [X,Y] + L_{X} g - L_Y f - F_2(X,Y) \, ,
$$
hence fitting in the short exact sequence of Lie algebroids over X
$$ \mathbb{R}_X \longhookrightarrow TX\oplus \mathbb{R}_X \longrightarrow TX \, .$$
In this explicit form, this Lie algebroid has been recently termed the ``Almeida-Molino Lie algebroid'' by \cite{Mein}, originally due to
\cite{AM85}. 

Lastly, note that any choice of gauge equivalent connection $\overline{e} = e' + \dd N_0$ yields precisely the same Lie algebroid, while in general, an inequivalent connection $\widetilde{e}_2$ with curvature $\widetilde{F}_2$ yields only a \textit{quasi-isomorphic} Lie algebroid, since both curvatures define the same (integral) de Rham cohomology class so that $\widetilde{F}_2 = F_2 + \dd C$ for \textit{exact} $2$-form $\dd C\in \Omega^{2}_{\mathrm{dR}}(X)$. This is, of course, as expected since all of the resulting Lie algebroids arise as duals of sCDGAs which are all isomorphic to the \textit{same} sCDGA $\Omega^{\bullet}_{\mathrm{dR}}(Y)_{\mathrm{inv.}}$.
\end{remark}

\section{
\texorpdfstring
{IIA SuGra via $\mathrm{cyc}(\mathfrak{l} S^4)$-valued cocycles on super-spacetime}
{IIA SuGra via cyc(lS4)-valued cocycles on super-spacetime}
}
Recall that a super-spacetime (\cite[Def. 2.74]{GSS24-SuGra}) of dimension $\big(1,d\, \vert \,  \mathbf{N}\big)$ is a supermanifold $Y$ locally isomorphic to $\mathbb{R}^{1,d\vert \mathbf{N}}$, that is supplied with a super Cartan connection (the super-gravitational field)
$
(e 
, \, \psi, \, \omega)\, 
$ with respect to the super-group inclusion $\mathrm{Spin}(1,d) \hookrightarrow  \mathrm{ISO}\big( \mathbb{R}^{1,d|\mathbf{N}}\big)$\,, which is furthermore (super-)torsion-free in the bosonic direction\footnote{Note that, in line with the discussion from Sec. \eqref{IntroductionSection}, this condition should also be thought of as a (canonical) curved global extension of the tangent-space cochain relation from Eq. \eqref{SuperMinkowskiCE}.}
\begin{align}\label{TorsionLessCondition}
T\, := \, \dd \, e + \omega \wedge e - \big(\, \overline{\psi}
    \,\Gamma\,
    \psi
    \big) \, \equiv \, 0 \, . 
\end{align}
Being a Cartan connection entails, in particular, that over local trivializations of the supermanifold $\{U_{i} \xrightarrow{\sim} \mathbb{R}^{1,d\vert \mathbf{N}}\}_{i\in I}$, the gravitational field is represented by the ``Spin connection'' gauge field 1-forms 
$$
\big\{ \, \omega_i \, \in \, \Omega^{1}_{\mathrm{dR}}\big(U_i ;\,  \mathfrak{so}(1,d)\big)  \big\}_{i\in I}
$$
and ``super vielbein'' 1-forms
$$
\big\{ (e_i, \, \psi_i) \, \in \, \Omega^{1}_\mathrm{dR}\big(U_i;\,  \mathbb{R}^{1,d|\mathbf{N}}\big)  \big\}_{i\in I} \, .
$$
These locally defined representatives are related on overlaps via $\mathrm{Spin}(1,d)$-valued transition functions $\{\gamma_{ij} : U_{ij} \longrightarrow \mathrm{Spin}(1,d) \}_{i,j\in I}$. The latter pair of $1$-forms is required to constitute a super-coframe, i.e., an isomorphism of super-tangent bundles 
$$
TU_i \xlongrightarrow{\sim} T \mathbb{R}^{1,d|\mathbf{N}}\, ,
$$
for all $i\in I$, in a compatible manner over overlaps, yielding a global isomorphism of super-tangent bundles\footnote{Here we follow the notation from \cite{Gi26}: For a background $G$-structure on $X$, presented as a principal $G$-bundle $P\rightarrow X$, the associated vector bundle corresponding to $G$-action on vector space $V$ is denoted by 
$V_G \, = \, P\times_G V \, .$ }
$$
(e,\, \psi) \, : \, TY \xrightarrow{\quad \sim \quad } \FR^{1,d\vert \mathbf{N}}_{\mathrm{Spin}(1,d)}
$$
to the corresponding associated vector bundle.
\medskip

\noindent {\bf Symmetric super-spacetimes.}
Here, we are interested in the dimensional reduction of $S^1$-symmetric super-spacetimes. The appropriate $\mathrm{Spin}(1,d)$-covariant isometry condition (hence manifestly globally well-defined) on such super-spacetimes is formulated in \cite[Def. 3.12]{Gi26}: Namely,  a super-spacetime $$
\big(Y,\,(e,\psi,\omega)\big)
$$ is (infinitesimally) $S^1$-symmetric with respect to some $S^1$-action on $Y$ if the \textit{Kosmann Lie derivative}\footnote{Explicitly, this is given by $L_\xi^K e = L_\xi e -  B_\xi\cdot e $ where $B$ is (locally) defined as the map $B_i: \CX(Y)\rightarrow \Omega^1\big(U_i; \mathfrak{so}(1,d)\big)$ given by the matrix elements $(B_\xi)^{a}{}_b := -\eta^{af} \, (L_\xi e)_{[bf]} $. The Kosmann Lie derivative on $\psi$ is given by the same formula, with $B$ acting via the corresponding $\mathrm{Spin}(1,d)$ representation, while on the gauge field $\omega$ it is given by  $L_\xi^K \omega = L_\xi \omega - \dd^\omega B_\xi$. See \cite{Gi26} for more on its geometrical origin, and for a comprehensive account of the theory of covariant Lie derivatives.} of the gravitational field, along the corresponding fundamental vector field, vanishes
$$
\big(L_\xi^K e, \, L_\xi^K \psi , \, L_\xi^K \omega \big) \, = \, 0 \, .
$$
The crucial aspect of this condition is that the vanishing $L_\xi^K e =0$ of the bosonic part is equivalent to the Killing equation $L_\xi g^{\even} := L_\xi ( \eta_{ab} \, e^a\otimes e^b) =0$ for the corresponding (bosonic) metric\footnote{Hence, in particular, lifting the traditional isometry condition on bosonic spacetimes to super-spacetimes.} \cite[Lem. 3.14]{Gi26}. Phrased via the \textit{covariant} Kosmann Lie derivatives, however, this condition makes sense in the first-order (coframe) formalism on arbitrary topologies, and in particular also on non-parallelizable super-spacetimes. Nevertheless, \cite[Cor. 3.16]{Gi26} exhibits that when the $S^1$-action is well-behaved, in that the super-spacetime Y forms a (super-) principal $S^1$-bundle, there always exists a  \textit{partial} gauge fixing of the fields, corresponding to a reduction of the structure group\footnote{Here we assume the action is spacelike. A timelike action has an analogous reduction of the corresponding signature. } $\mathrm{Spin}(1,d-1)\hookrightarrow \mathrm{Spin}(1,d)$,
$$
(e,\psi,\om) \xrightarrow{\quad r \quad } \big(\hat{e},\hat{\psi},\hat{\om}\big) : = r\cdot (e,\psi,\om)
$$
in which the Kosmann vanishing condition of the given field configuration, and of its whole  Spin$(1,d-1)$-gauge equivalence class, \textit{coincides} with the vanishing of their naive Lie derivative
\begin{align}\label{NaiveLieDerivativeVanishing}
\big(L_\xi \hat{e}, L_\xi \hat{\psi},L_\xi \hat{\om}\big)= 0\, .
\end{align}

Henceforth, we will assume the super-spacetime $Y$ has a principal $S^1$-bundle structure, and further that its $S^1$-symmetric gravitational field is already (partially) gauge-fixed so that its naive Lie derivative vanishes. Nevertheless, with that understood, we will use the un-hatted notation for brevity. For the sake of being further explicit with the physical application at hand, we will also fix the bosonic dimension to be $D=1+10=11$ and the corresponding spinorial representation to be $\mathbf{32} \, \in \mathrm{Rep}_\FR\big(\mathrm{Spin}{(1,10)}\big)$, hence branching as a direct sum of ``left and right'' representations (see, e.g., \cite[\S 3.1]{GSS24-TDuality}) $\mathbf{16}\oplus \overline{\mathbf{16}} \, \in \mathrm{Rep}_\FR\big(\mathrm{Spin}{(1,9)}\big)$ under the reduction of the structure group $\mathrm{Spin}(1,9)\hookrightarrow \mathrm{Spin}(1,10)$ induced by the partial gauge fixing from  \cite[Cor. 3.16]{Gi26}.

\medskip 
\noindent {\bf The torsion condition in the partially fixed gauge.}
Towards dimensional reduction, we are interested in the decomposition of the (bosonic) torsion-less condition \eqref{TorsionLessCondition} in the (partially) fixed gauge where the symmetry condition corresponds to the naive Lie derivative vanishing \eqref{NaiveLieDerivativeVanishing}. Recall from \cite[Cor. 3.16]{Gi26} that in this situation, the bosonic component of the coframe is decomposed as 
\begin{align}\label{11DPartialGaugeFix}
e_Y \,=\, \big(e_X, \, e^{\ten}\big) \,,
\end{align}
where $e_X$ is a basic field\footnote{We will abuse notation and denote basic (i.e., $S^1$-invariant and horizontal) form fields on $Y$ by the same notation as their counterparts on $X$, hence omitting the pullback symbol $\pi^*$.}, the pullback of a (bosonic) coframe on the base of the principal $S^1$-bundle $X^{10}:= Y^{11} / S^1$, and $e^{\ten}$ is a $1$-form that yields a global parallelism
of the vertical tangent bundle $$
e^{\ten} \, : \, VY \xlongrightarrow{\sim} Y\times \FR 
$$
where $\FR\cong i \, \FR \cong \mathfrak{u}(1)$ is the Lie algebra of the corresponding $U(1)\cong S^1$ group. Since
$$
\iota_\xi e^{\ten}\, = \, \Phi \;\; \in \;\; C^\infty_{\even}(Y)\,,
$$
for a unique (non-vanishing, even) $U(1)$-invariant function $\Phi$, where $\xi$ is the corresponding fundamental vector field, 
it follows that $e^{10}$ is proportional to a (unique) $U(1)$-connection $e' \in \mathrm{Conn}_{U(1)}(Y)$
\begin{align}\label{DilatonConnectionRelation}
e^{\ten} \, = \, \Phi \, e' \, , 
\end{align} 
whose (basic) \textit{curvature} we denote by
$$
F_2 \, = \, \dd e' \; \in \; \Omega^2_\even(Y)\, ,
$$
being the pullback of a (globally defined and integral) 2-form $F_2$ on $X$.

Crucially, extracting the conditions with respect to the base $X^{10}$ requires decomposing also the fermionic coframe in terms of basic fields. That is, even though the fermionic coframe $\psi_Y$ is $C^\infty(Y)$-linearly independent\footnote{In that their components jointly form a basis for the super-vector space ($C^\infty(Y)$-module) of differential $1$-forms on $Y$.} of the bosonic part $e_Y$, it is possible that $\psi_Y$ has a non-trivial expansion in the direction of the connection $e'$, i.e.,
\begin{align}\label{11DPartialGaugeFixFermion}
\psi_Y \, = \, \psi_X\, + \, \lambda \, e'
\end{align}
for some unique $S^1$-invariant \textit{fermionic} $1$-form $\psi_X$ and $0$-form $\lambda$, both (locally) valued in the branched representation $\mathbf{16}\oplus \overline{\mathbf{16}} \, \in \mathrm{Rep}_\FR\big(\mathrm{Spin}{(1,9)}\big)$.\footnote{One can split these according to their irreducible components, e.g., $\psi= \psi_L + \psi_R \, \in \,  \mathbf{16}\oplus \overline{\mathbf{16}}$. We do not do this explicitly to keep the formulas more concise.} Notice, in analogy to the fact that $\iota_\xi e^{\ten} = \Phi$, this means in particular that\footnote{The first version of \cite{Gi26} states the $S^1$-symmetric condition only for even vector fields with no legs along the fermionic frame (cf. ftn. of Def. 3.13 therein). However, this is a needless loss of generality, and the same definition is applicable for any even vector field. That is $\xi = \xi^a e_a + \xi^\beta \psi_\beta  \, \in \, \CX_\even(Y)$, so that $\xi^a \in C^\infty_\even(Y)$ and $\xi^\beta \in C^\infty_\odd(Y).$}
$$
\iota_\xi \psi_Y \, = \, \lambda \, .
$$
Just as with the curvature $2$-form, both $\psi_X$ and $\lambda$ are  $S^1$-invariant and \textit{horizontal}, hence descending to fields on the base. 

Overall, $S^1$-invariant super-coframes $(e_Y, \psi_Y)$ on $Y$, in the partially fixed gauge of \cite[Cor. 3.16]{Gi26}, are in $1$-$1$ correspondence with super-coframes 
\begin{align}\label{BasicSuperCoframe}
(e_X,\,  \psi_X) \, : \, TX \xrightarrow{\quad \sim \quad} \FR^{1,9\vert \mathbf{16}\oplus \mathbf{\overline{16}}}_{\mathrm{Spin}(1,9)}
\end{align}
on the base $X$, together with a pair of fields
\begin{align}\label{DilatonDilatinoPair}
(\Phi, \, \lambda) \quad \in \quad C^\infty_\even(X) \times C^\infty_\odd\big(X,\,(  \mathbf{16}\oplus \overline{\mathbf{16}})_{\mathrm{Spin(1,9)}}\big)\, ,   
\end{align}
i.e., the superized \textit{dilaton} and \textit{dilatino} ``superpartners'', and a (locally defined) $U(1)$-gauge field $A$ with curvature
$$F_2 \, \in \, \Omega^2_\even(X)\, .$$

Lastly, taking into account the (locally defined) $S^1$-invariant Spin$(1,10)$-gauge field $\omega_Y$ on $Y$, its matrix components are decomposed, similarly, as
\begin{align}\label{11DPartialGaugeFixConnection}
\omega_Y^{ab} \,  = \, \omega_X^{ab} \, + \, \widetilde{\omega}^{ab}_X \, e'
\end{align}
for some (locally defined) invariant $S^1$-invariant and \textit{horizontal} $1$-forms and $0$-forms valued in $\mathfrak{spin}(1,10)\cong \mathfrak{so}(1,10)$, respectively. Since the partial gauge fixing reduces the structure group to that of an $S^1$-invariant\footnote{In that the corresponding reduced transition functions are also $S^1$-invariant (via the integration action) due to  \eqref{NaiveLieDerivativeVanishing}.} $\mathrm{Spin}(1,9)\hookrightarrow \mathrm{Spin}(1,10)$, it follows that the subset of (locally defined) matrix
components 
\begin{align*}
\{\omega_{X}^{ab}\}_{a,b=0,\cdots,9}
\end{align*}
defines a $\mathrm{Spin}(1,9)$-gauge field $\omega_X$ on the base supermanifold $X$. Hence, the triplet 
$$
(e_X,\, \psi_X,\, \omega_X)
$$
forms a Cartan connection on $X$ with respect to the super-group inclusion $\mathrm{Spin}(1,9)\hookrightarrow \mathrm{ISO}(\FR^{1,9\vert \mathbf{16}\oplus \overline{\mathbf{16}}})$.

\smallskip 
With the above details spelled out, the 11D torsion constraint \eqref{TorsionLessCondition} decomposes as follows.
\begin{lemma}[\bf 11D Torsion constraint in terms of basic 10D fields]\label{11DTorsionContraintIn10D}
When the $(1+10)$-bosonic super-torsion constraint $\dd \, e_Y + \omega_Y\wedge e_Y - \big(\,\overline{\psi}_Y 
    \,\Gamma\,
    \psi_Y
    \big) = 0 $ \eqref{TorsionLessCondition} is expanded in the symmetric partially fixed gauge \eqref{NaiveLieDerivativeVanishing} with super-coframe and Spin-connection basic expansions \eqref{11DPartialGaugeFix} \eqref{11DPartialGaugeFixFermion} and \eqref{11DPartialGaugeFixConnection}, respectively, it may be expressed as the following list of conditions:
\begin{itemize}
\item[\bf (i)] The $(1+9)$-bosonic super-torsion constraint
\begin{align}\label{10DTorsionConstraint}
\dd \, e_X + \omega_X\wedge e_X - \big(\,\overline{\psi}_X 
    \,\Gamma\,
    \psi_X
    \big) \, = \, 0\,.
\end{align}
\item[\bf (ii)]  The super-coframe expansion of the $U(1)$-curvature $2$-form flux $F_2$ on $X$ as
\begin{align}\label{2FluxExpansion} 
F_2 \, = \, - \frac{1}{\Phi}\, \omega_{X\, [b,}
{}_{c]}{}^{\ten} \, e_X^b \,e_X^c \, + \, \frac{2}{\Phi^2} \big(\,\overline{\psi}_X \, \Gamma_{b} \, \lambda \big)\, e_X^b \, + \, \frac{1}{\Phi}\, \big(\,\overline{\psi}_X 
    \,\Gamma^{\ten}\,
    \psi_X
    \big) \,.
\end{align}

    \item[\bf (iii)]  The super-coframe expansion of the dilaton derivative $\dd \Phi$ on $X$ as
\begin{align}\label{10DDilatonDifferrential}
 \dd \Phi \, = \, \widetilde{\omega}_X{}^{\ten}{}_a\, e^a_X \, + \, 2 \big(\,\overline{\psi}_X 
    \,\Gamma^{\ten}\,
    \lambda
    \big)    \,.
\end{align}

\item[\bf (iv)]  The following prescription for the reconstruction of the remaining parts of the 11D connection, as
$$
\widetilde{\omega}_X{}^{a}{}_b \, =\, \Phi\, \omega_{X\, b,}{}^{a\ten} \;,  \quad \qquad \omega_{X\, \beta,}{}^{a\ten} \,\psi^\beta_X = \, \frac{2}{\Phi}\, \big(\,\overline{\psi}_X\, \Gamma^a\, \lambda \big)\,,
$$
for $a,b =1,\cdots\, 9$, which can be further expressed as functions of the components of $F_2$ via \eqref{2FluxExpansion}.
\end{itemize}
\end{lemma}
\begin{proof}
This is a standard calculation that follows by a careful component expansion 
\begin{align*}
T\,  &=  \, \tfrac{1}{2} T_{bc}\,  e^b_Y\, e^c_Y + T_{b\gamma} \, e^b_Y\, \psi^{\gamma}_Y + \tfrac{1}{2}T_{\beta \gamma} \,\psi^{\beta}_Y \, \psi^{\gamma}_Y
\end{align*} 
of the condition $T\, := \, \dd \, e_Y + \omega_Y \wedge e_Y - \big(\, \overline{\psi}_Y
    \,\Gamma\,
    \psi_Y
    \big) \, \equiv \, 0$, in the given partially fixed gauge with respect to the super-coframe and Spin-connection basic expansions \eqref{11DPartialGaugeFix}, \eqref{11DPartialGaugeFixFermion}, and \eqref{11DPartialGaugeFixConnection}. For instance, consider the $(\psi_X\mbox{-}e')$ leg of the above condition, say along the $\ten$-th direction, i.e., 
    \begin{align*}
    2\, \big(\,\overline{\psi}_X 
    \,\Gamma^{\ten}\,
    \lambda
    \big) \, e'   \, &= \, \dd e_{\beta \, \ten ,}{}^{\ten} \, \psi_X^\beta\, \Phi \, e' + \omega_{\beta,}{}^{\ten}\,_{\ten} \, \psi_X^\beta \, \Phi \, e' \, =\,  \dd (\Phi \cdot e')_{\beta \, \ten } \, \Phi \, \psi_X^\beta \, e' \, +\,  0 \\
    &= \,  \big(\dd \Phi_\beta \, (e')_{\ten} + \Phi \, F_{2 \, \beta \ten} \, \big) \, \Phi \, \psi_X^\beta \, e' \, = \, \dd \Phi_\beta \, \frac{1}{\Phi} \, \Phi \, \psi_X^\beta \, e' 
    \\
    &= \, \dd \Phi_\beta \, \psi_X^\beta \, e'  
    \end{align*}
    which yields the fermionic component of \eqref{10DDilatonDifferrential}. Notice, we used the facts that by definition $\omega^{\ten\, \ten}=0$, $F_{2\, \beta \ten} =0 $ since it is basic along $Y^{11}\rightarrow X^{10}$, and $(e')_{\ten}=\frac{1}{\Phi}$ since $e'=\tfrac{1}{\Phi}\,e^{\ten}$. The rest of the results follow via similar considerations.
\end{proof}
This decomposition is, naturally, already well-established on product super-spacetimes (see, e.g., \cite[Appendix C]{HS05} -- up to differing conventions), and now it is further \textit{justifiably} applicable on arbitrary $S^1$-principal bundle super-spacetimes.


\paragraph{Convention.}
    For ease of notation, from here onwards we will drop the $X$ subscript on the basic fields, and keep only the total space $Y$ subscript where necessary.

\medskip 
\noindent 
{\bf Supersymmetry transformations via the 11D super-torsion constraint.}
Recall that one of the positive attributes of the superspace formulation of (on-shell) supergravity is the manifest interpretation of supersymmetry transformations as the action of (infinitesimal) fermionic diffeomorphisms (see, e.g., \cite{CDF91}\cite{Gi26}), that is, the (accordingly) covariantized action of odd vector fields $\eta \in \CX^\odd(X)$. In particular, 
the decomposed bosonic torsion constraint immediately yields the expected form of the IIA 10D (on-shell) supersymmetry transformation on the bosonic coframe field:
\begin{align*}
 L_\eta^{\omega_X} \,e_X  &\equiv \, [\iota_\eta, \dd^{\omega_X}\,] e_X \, = \, \iota_\eta \dd^{\omega_X}\, e_X \, \\& = \,  2(\,\overline{\epsilon}_{\eta}\, \Gamma \, \psi_X )\, =: \, \delta^{10\mathrm{D}\mbox{-}\mathrm{susy}}_\eta e_X\, , 
\end{align*}
which indeed coincides with the expected transformation upon restriction to the underlying 10-dimensional bosonic spacetime $\bosonic{X}\hookrightarrow X$. From the remaining induced conditions \eqref{2FluxExpansion} and \eqref{10DDilatonDifferrential}, one may immediately ``re-derive'' the corresponding IIA 10D supersymmetry transformations via super-spacetime.

\begin{corollary}[\bf 2-Flux and Dilaton supersymmetry via superspace]\label{2FluxAndDilatonSupersymmetry}
Under the decomposition of the partially gauge-fixed 11D Torsion constraint of Lem. \ref{11DTorsionContraintIn10D}, the actions\footnote{Since $\Phi$ and $F_2$ are both valued in the \textit{trivial} $\mathrm{Spin}(1,9)$ representation, the covariant Lie derivative reduces to the traditional Lie derivative.} of 10D odd vector fields $\eta \in \CX^\odd(X)$ on the (superized) $2$-Flux $F_2$ and the dilaton field $\Phi$ have the form -- up to conventions\footnote{And further field redefinitions extrapolating between the ``Einstein and String frames'' of IIA 10D supergravity.} -- of the corresponding IIA 10D supersymmetry transformations:
\begin{align*}
L_\eta \, F_2 
&= \, 2 \, \dd \, \big( \tfrac{1}{\Phi^2}(\overline{\epsilon}_\eta \, \Gamma_b \, \lambda)\, e_X^b \, + \, \tfrac{1}{\Phi} (\overline{\epsilon}_\eta \, \Gamma^{\ten} \, \psi_X )  \big)
\\
L_\eta \, \Phi \, 
&= \, 2 \big( \overline{\epsilon}_\eta \, \Gamma^{\ten} \, \lambda \big) \, . 
\end{align*}

\end{corollary}

\begin{proof}
    This follows directly from using equations \eqref{2FluxExpansion} and \eqref{10DDilatonDifferrential}, respectively, in 
$$
L_\eta \, F_2 \, \equiv \,   [\iota_\eta, \dd]\, F_2 \, = \,  \dd \, \iota_\eta F_2 \quad {\rm and} \quad  
L_\eta \, \Phi \, \equiv \,  [\iota_\eta, \dd]\, \Phi \, = \, \iota_\eta \dd \Phi \,.
$$

\vspace{-4mm}
\end{proof}

\begin{remark}[\bf Dilaton supersymmetry via 11D superspace]
$\,$
\begin{itemize}[leftmargin=.8cm]
\item[\bf (i)] Recall that in the original formulation of IIA 10D supergravity on \textit{bosonic spacetimes} 
\cite{CW84}\cite{GP84} \cite{HN85}, the derivation of the 10D  supersymmetry transformation was not as immediate. Namely, these were derived by considering the subset of 11D infinitesimal transformations (i.e., local vector fields on 11D \textit{field space} $\mathcal{F}^{11D}$ \cite{GS25}) that descend to the 10D field space $\mathcal{F}^{10D}$, hence those that preserve the (partial) gauge fixing from \eqref{11DPartialGaugeFix}. In the original terminology, these are the ones that preserve the ``triangular'' form of the bosonic coframe. 

\item[\bf (ii)] For expositional purposes, we express this equivalent way of deriving the 10D supersymmetry transformations, but now on \textit{super-spacetime} so as to relate to the action of 11D odd vector fields. To do this, instead of expressing the 11D torsion constraint in terms of basic fields as in Lem. \ref{11DTorsionContraintIn10D}, one expands it in the gauge-fixed 11D super-coframe basis $(e_Y,\psi_Y)$ with $e_Y= (e_X, \Phi \, e')$ but without decomposing the 11D fermionic coframe $\psi_Y$ in terms of the connection $e'$. This allows us to compute directly the action of odd vector fields on the 11D spacetime $Y^{11}$, rather than the base $X^{10}$, hence encoding the (on-shell) 11D supersymmetry transformations. For example, computing the fermionic components of the derivative of the dilaton, but now instead with respect to $\psi_Y$, it follows \footnote{Explicitly, the $e^{\ten}\mbox{-}\psi_Y^{\overline{\gamma}}$ leg of the 11D Torsion constraint yields instead 
\vspace{-3mm} 
$$
    0 \, = \, \tfrac{1}{2} \dd e _{\overline{\gamma}\, \ten }{}^{\ten} + \omega_{\overline{\gamma},}{}^{\ten}{}_{\ten} \, =\, \tfrac{1}{2} \dd (\Phi \, e')_{ \overline{\gamma}\, \ten} + 0 \, = \, \tfrac{1}{2}\dd \Phi_{\overline{\gamma}} \,  e'_{\ten} \, + \, \tfrac{1}{2}\Phi \, F_2{}_{\overline{\gamma}\, \ten }  
    =  (\dd \Phi)_{\overline{\gamma}} \, \frac{1}{\Phi} \, . 
    $$}
    that 
$$
(\dd \Phi)_{\overline{\gamma}} \, \equiv \, \partial_{\overline{\gamma}} \Phi \, = \, 0 \, ,
$$
where we used overlined Greek indices $\{\overline{\alpha}, \overline{\beta}, \cdots\}$ to distinguish them from those in the further basic 

\newpage 
expansion in Lem. \ref{11DTorsionContraintIn10D}. 
This immediately implies that on $Y^{11}$, the action of \textit{any} odd vector field $\widehat{\eta}\in \CX^\odd(Y)$ on the dilaton is trivial 
$$
L_{\widehat{\eta}} \,\Phi \, = \, 0 \, ,
$$
even for $S^1$-invariant (projectable) vector fields, which says that the canonical 11D supersymmetry transformations act trivially on the (11D incarnation) of the dilaton field. 

\item[\bf (iii)] Of course, the conceptual resolution is not hard to see: Even if $\widehat{\eta}$ projects to a vector field $\eta$ on $X^{10}$, the corresponding vector field on 11D field space $\mathcal{F}^{11D}$ given by the corresponding Lie derivative does not descend to one on 10D field space $\mathcal{F}^{10D}$. In simple terms, this is due to the fact that the standard 11D susy transformation does not preserve the partially fixed gauge from \eqref{11DPartialGaugeFix}. Hence, precisely as done in the original papers for bosonic spacetimes \cite{CW84}\cite{GP84}\cite{HN85}, one should compensate any such 11D susy transformation, generated by  some $S^1$-invariant odd vector field $\widehat{\eta}\in \CX^{\odd}(Y)$ projecting to $\eta \in \CX^\odd(X)$, with a (bi-fermionic and field-dependent) 11D local Lorentz transformation so that
$$
\big(\hat{\delta}^{11\mathrm{D}\mbox{-}\mathrm{susy}}_{\widehat{\eta}}  e^{a}\big)_{\ten}\, : =  \big(\delta^{11\mathrm{D}\mbox{-}\mathrm{susy}}_{\widehat{\eta}} e^a \big)_{\ten} \, + \, \big(\delta_\Lambda^{\mathrm{Spin}(1,10)} e^a\big)_{\ten} \, \equiv \, 0 \, 
$$
 for $a=0,\cdots, 9$. It is readily seen that this is achieved by choosing the (local) rotation to have matrix elements 
 $$
 (\Lambda_{\widehat{\eta}, \xi})^{a \ten} \, := \, - 2 \big( \overline{\epsilon}_{\widehat{\eta}} \, \Gamma^a \, \iota_\xi \psi \big) \, \equiv \, -  2 \big( \overline{\epsilon}_{\widehat{\eta}} \, \Gamma^a \, \lambda\big)
 $$
 and $\Lambda_{\eta,\xi}^{ab}=0$ for $a,b =1, \cdots , 9$. 

\item[\bf (iv)] Finally, acting with this new adjusted (and still projectable) infinitesimal 11D ``supersymmetry'' transformation on $e^{10}= \Phi\, e'$,
\vspace{1mm} 
$$
\hat{\delta}^{11\mathrm{D}\mbox{-}\mathrm{susy}}_{\widehat{\eta}} e^{\ten} \, = \, \big(\hat{\delta}^{11\mathrm{D}\mbox{-}\mathrm{susy}}_{\widehat{\eta}} \Phi\big) \, e' + \, \Phi \, \big(\hat{\delta}^{11\mathrm{D}\mbox{-}\mathrm{susy}}_{\widehat{\eta}} e'\big)  
$$
 and expanding the left-hand side as above, one finds in particular a non-zero (projectable) supersymmetry transformation for the dilaton field
$$
\hat{\delta}^{\rm susy}_{\widehat{\eta}}\Phi \, = \, 2 \big(\overline{\epsilon}_{\widehat{\eta}} \, \Gamma^{\ten}\, \lambda\big) \, , 
$$
which, since $\widehat{\eta}$ projects to $\eta \in \CX^{\odd}(X)$, descends precisely to the corresponding 10D supersymmetry transformation of the dilaton field from Cor. \ref{2FluxAndDilatonSupersymmetry}. 
\end{itemize}
\end{remark}

\noindent {\bf Superspace dimensional reduction of $S^4$-flux quantizable 11D supergravity.}
Recall \cite[Thm. 3.1]{GSS24-SuGra} that the condition that a pair of super-fluxes of the form 
of the following form:
\begin{equation}
\label{SuperFluxDensitiesInIntroduction}
  \def\arraystretch{1.5}
  \begin{array}{l}
    G^s_4
    \;\;:=\;\;
    \tfrac{1}{4!} (G_4)_{a_1 \cdots a_4}
    e^{a_1}_Y \cdots e^{a_4}_Y
    \;\,+\;\,
    \tfrac{1}{2}
    \big(\,
    \overline{\psi}_Y
    \Gamma_{a_1 a_2}
    \psi_Y
    \big)
    e^{a_1}_Y \, e^{a_2}_Y
    \\
    G^s_7
    \;\;:=\;\; 
    \tfrac{1}{7!}
    (G_7)_{a_1 \cdots a_7}
    e^{a_1}_Y \cdots e^{a_7}_Y
    \;\,+\;\,
    \tfrac{1}{5!}
    \big(\,
    \overline{\psi}_Y
    \Gamma_{a_1 \cdots a_5}
    \psi_Y
    \big)
    e^{a_1}_Y \cdots e^{a_5}_Y
  \end{array}
\end{equation}
constitutes a map of $L_\infty$-algebroids 
$$
(G_4^s, \, G_7^s) \, : \, TY^{11}\longrightarrow \mathfrak{l}S^4 \, ,
$$ over a $(11 \vert \mathbf{32})$-dimensional super-spacetime $(Y^{11},e,\omega, \psi)$ modeled on $\FR^{1,10\vert \mathbf{32}}$.
This means that satisfying  the $\mathfrak{l}S^4$-cocycle condition
\begin{equation}
\label{SuperCFieldBianchiInIntro}
  \def\arraystretch{1.2}
  \begin{array}{l}
    \differential
    \, 
    G^s_4
    \;=\; 0
    \\
    \differential
    \, 
    G^s_7
    \;=\;
    \tfrac{1}{2}
    G^s_4 \wedge G^s_4
    \,,
  \end{array}
\end{equation}
is \textit{equivalent} to the tuple 
$$
\big( Y^{11},\, (e_Y,\psi_Y,\omega_Y,G_4,G_7)\big)
$$ 
being a solution of the on-shell equations of \textit{superspace} 11D supergravity, which is furthermore the \textit{unique rheonomic extension} of its restriction to the solution of 11D supergravity on the underlying bosonic manifold $\bosonic{Y}^{11}\hookrightarrow Y^{11}$ \cite[Cor. 3.12]{GSS24-SuGra}. 
\begin{equation}
  \label{11DTabledStatement}
  \hspace{-3mm} 
 \colorbox{lightgray}
  {\!\!\!\!\!
     \def\arraystretch{1.1}
    \begin{tabular}{c}
      $(11\vert\mathbf{32})$-dimensional 
      super-spacetimes
      $\big(
        Y, (e_Y, \psi_Y, \omega_Y)
      \big)$
      \\
      carrying super-flux 
      $(G_4^s, G_7^s)$
      (from \eqref{SuperFluxDensitiesInIntroduction}),
      \\
      satisfying the $\mathfrak{l}S^4$ Bianchi identity
      (from \eqref{SuperCFieldBianchiInIntro}).
    \end{tabular}
  \!\!\!}
  \hspace{.3cm}
  \Leftrightarrow
  \hspace{.3cm}
  \colorbox{lightgray}{\!\!\!
    \begin{tabular}{c}
      Solutions of 11D SuGra on $\bosonic{Y}$
      \\
        with flux source $G_4$
      \\
      and dual flux $G_7$.
    \end{tabular}
    \!\!\!}
\end{equation}

\smallskip 
Considering the subspace of $S^1$-symmetric (as in \cite[Def. 3.13]{Gi26}) 11D supergravity solutions in the above formulation, partially 
gauge-fixed as per \cite[Cor. 3.16]{Gi26} combined with Lem. \ref{11DTorsionContraintIn10D}, and applying the double dimensional reduction 
via cyclification of Prop. \ref{CyclificationOxidationOnSuperManifolds}
immediately yields the following superspace formulation of 10D IIA supergravity. 

\smallskip 
\begin{theorem}[\bf 10D IIA SuGra EoM from super-flux Bianchi identity]
\label{10dIIASugraEoMFromSuperFluxBianchiIdentity}
$\,$

\hspace{-6.5mm} 
\colorbox{lightblue}{\!\!\!
\begin{tabular}{l}
\begin{minipage}{17cm}
A $(10\vert\mathbf{16}\oplus \overline{\mathbf{16}})$-dimensional super-spacetime $\big(X^{10}, (e,\psi,\omega)  \big)$  \textup{(according to \cite[Def. 2.74]{GSS24-SuGra})}
carries super-field fluxes 
of the form 
\begin{equation}
\label{IIASuperFluxDensities}
  \def\arraystretch{1.5}
  \begin{array}{l}
    F^s_2
    \;\;:=\;\;
    \tfrac{1}{2}(F_2)_{a_1 a_2} \, e^{a_1} e^{a_2} \;\, + \, \, \frac{2}{\Phi^2} (\overline{\psi}  \Gamma_{a} \lambda)\, e^a \,\, +\;\, \frac{1}{\Phi} (\overline{\psi} \Gamma_{\ten} \psi)
    \\
        H^s_3
    \;\;:=\;\;  \tfrac{1}{3!}(H_{3})_{a_1 a_2 a_3} \, e^{a_1}e^{a_2}e^{a_3} \,\, + \, \, 
    (
    \overline{\psi}
    \Gamma_{a_1 a_2}
    \lambda )\,
    e^{a_1}  e^{a_2} \, \,  \;\,- \;\, \Phi (\overline{\psi} \Gamma_{\ten a} \psi)\, e^a 
    \\
        F^s_4
    \;\;:=\;\;
    \frac{1}{4!}(F_4)_{a_1 \cdots a_4} \, e^{a_1}\cdots e^{a_4} \;\, 
    \hspace{3.7cm} 
    +\;\, \tfrac{1}{2}(\overline{\psi} \Gamma_{a_1 a_2} \psi)\,  e^{a_1} e^{a_2} 
    \\
        F^s_6
    \;\;:=\;\; 
    \frac{1}{6!}(F_6)_{a_1 \cdots a_6}\,  e^{a_1}\cdots e^{a_6} \;\,+\;\,
    \tfrac{2}{5!}
    (
    \overline{\psi}
    \Gamma_{a_1 \cdots a_5}
    \lambda
    ) \, 
    e^{a_1} \cdots e^{a_5} \;\,-\;\, \tfrac{\Phi}{4!}(\overline{\psi} \Gamma_{\ten \, a_1 \cdots a_4} \psi) \,  e^{a_1}\cdots e^{a_4}  
    \\
    H^s_7
    \;\;:=\;\; 
    \tfrac{1}{7!}
    (H_7)_{a_1 \cdots a_7}\, 
    e^{a_1} \cdots e^{a_7}
    \;\,   
    \hspace{4.55cm} 
    + \;\,
    \tfrac{1}{5!}
    (
    \overline{\psi}
    \Gamma_{a_1 \cdots a_5}
    \psi
    ) \, 
    e^{a_1} \cdots e^{a_5} \, ,
  \end{array}
\end{equation}
where $\Phi\neq 0 \in C^\infty_\even(X)$ and $\lambda \in C^\infty_\odd\big(X,\,(  \mathbf{16}\oplus \overline{\mathbf{16}})_{\mathrm{Spin(1,9)}}\big)$ with 
\begin{equation}\label{DilatonDilatinoEquation}
 \dd \Phi \, = \, \partial_a\Phi\, e^a \, + \, 2 \big(\,\overline{\psi} 
    \,\Gamma^{\ten}\,
    \lambda
    \big)  \, 
\end{equation}
and integral 2-flux $[F_2^s] \, \in H^2_\mathrm{dR}(X,\mathbb{Z})$,  that constitute a map of $L_\infty$-algebroids 
$$
\big(F_2^s, \, H_3^s, \, F_4^s, \,  F_6^s, \, H_7^s\big)\, : \, TX^{10}\longrightarrow \mathrm{cyc}(\mathfrak{l}S^4) \, .
$$
That is, the fluxes form a closed $\mathrm{cyc}(\mathfrak{l}S^4)$-valued form
\begin{equation}
\label{Cyc-valuedForm}
\big(F_2^s, \, H_3^s, \, F_4^s, \,  F_6^s, \, H_7^s\big) 
  \,\in\, 
  \Omega^1_{\mathrm{dR}}\big(X^{10};\, \mathrm{cyc}(\mathfrak{l}S^4)\big)_{\mathrm{clsd}} \, ,
  \end{equation}
in that they   satisfy the Bianchi identities
\begin{eqnarray}
\label{CycS4Bianchis}
\dd F_2^s  &=&  0 \;,   \qquad \dd F_4^s \, = \, H_3^s \wedge F_2^s\;, \qquad \dd F_6^s \, = \, - H_3^s \wedge F_4^s \;, 
\nonumber 
\\[-2pt]
\quad \dd H_3^s &=& 0\;, 
 \qquad \dd H_7^s \, = \, \tfrac{1}{2} F_4^s \wedge F_4^s \, + \, F_2^s \wedge F_6^s 
\end{eqnarray}
if and only if
\begin{itemize}[leftmargin=.6cm]
\item[\bf (i)]
the form \eqref{Cyc-valuedForm} solves the equations of motion of IIA 10D supergravity 
with the given $(F_2,H_3,F_4)$-flux sources and their duals;
\item[\bf (ii)]
the super-fields form a unique (``rheonomic'') extension of their restriction 
to the bosonic body spacetime $\bosonic{X}$.
\end{itemize}
\end{minipage}
\end{tabular}
\!\!}
\end{theorem}

\begin{proof}
This follows by a combination of several established results, summarized in the following diagram of super smooth sets\footnote{For simplicity of exposition, we indicate these here using the set notation, but one should keep in mind that these are really arbitrary $\FR^{k|q}$-plots of the corresponding super smooth sets (cf. developments in 
\cite{Gi24}\cite{GSS24-SuGra}\cite{SS-Zagreb}).}
of field spaces
\begin{equation}
  \begin{tikzcd}[
    row sep=30pt, column sep=60pt
  ]
    \left\{
    \substack{
      \text{$S^1$-symm. $\mathfrak{l}S^4$-cocycles}
      \\
      \text{$(G^s_4, G^s_7)$ on $Y^{11} \to X^{10}$}
    }
    \right\}
    \ar[
      r,
      <-,
      shift left=4pt,
      "{ \text{oxidation} }"
    ]
    \ar[
      r,
      shift right=4pt,
      "{ \text{cyclification} }"'
    ]
    \ar[
      d,
      <->,
      "{ 
        \text{rheonomy}
      }"{description}
    ]
    &
    \left\{
    \substack{
      \text{$\mathrm{cyc}(\mathfrak{l}S^4)$-cocycles}
      \\
      \text{$(F^s_2, H^s_3, F^s_4, F^s_6, H^s_7)$ on $X^{10}$}
    }
    \right\}
    \ar[
      d,
      <->,
      "{ 
        \text{rheonomy}
      }"{description}
    ]
    \\
    \left\{
    \substack{
      \text{$S^1$-symm. 11D SuGra}
      \\
      \text{sols on $\overset{\rightsquigarrow}{Y}{}^{11} \to \overset{\rightsquigarrow}{X}{}^{10}$ }
    }
    \right\}
    \ar[
      r,
      shift right=4pt,
      "{ \text{reduction} }"'
    ]
    \ar[
      r,
      <-,
      shift left=4pt,
      "{ \text{extension} }"
    ]
    &
    \Big\{
    \substack{
      \text{10D SuGra sols on $\overset{\rightsquigarrow}{X}{}^{10}$}
    }
    \Big\}
  \end{tikzcd}
\end{equation}
This is to be read as a clockwise composition from the bottom right. First, on-shell 10D IIA supergravity on a usual bosonic spacetime $\bosonic{X}^{10}$, at the level of flux curvature forms (i.e., ``locally''), is known to correspond precisely to the $S^1$-invariant (\cite[Def. 3.13]{Gi26}) sector of on-shell 11D supergravity, with principal $S^1$-bundle spacetime topology $\bosonic{Y}^{11}\!\rightarrow \!\bosonic{X}^{10}$ classified by the Chern class of the $2$-form flux on $\bosonic{X}^{10}$. In fact, this is precisely how IIA supergravity was initially derived \cite{CW84}\cite{GP84}\cite{HN85}, with the non-trivial $S^1$-bundle extensions in the second-order formalism found in \cite{FOS03}\cite{MaS04}\cite{Jose01}. With the proper formulation of first-order (coframe) $S^1$-invariance condition at hand \cite[Def. 3.13]{Gi26}, and the existence of the appropriate gauge for any non-trivial $S^1$-principal bundle spacetime \cite[Cor. 3.16]{Gi26}, this follows similarly also in this setting. In other words, there exists an isomorphism between the full (curvature flux-form) field space of on-shell 10D IIA supergravity and the field space of $S^1$-invariant on-shell 11D supergravity. 

Next, by \cite[Thm. 3.1]{GSS24-SuGra}, the latter field space is \textit{rheonomically}\footnote{Namely, to the canonically associated supermanifold extension $Y^{11} = \bosonic{Y}^{11} \vert \mathbf{32}_{\mathrm{Spin}(1,10)}$, constructed as the (split) supermanifold dual to the exterior algebra of sections of the associated vector bundle $\mathbf{32}_{\mathrm{Spin}(1,10)} \equiv \mathbf{32} \times _{\mathrm{Spin}(1,10)} P $ with respect to the background principal Spin$(1,10)$-bundle structure \cite[Ex. 2.77]{GSS24-SuGra}.} isomorphic to the space of $(11\vert \mathbf{32})$-dimensional $S^1$-symmetric superspacetimes $\big(Y^{11},\,  (e_Y,\psi_Y,\omega_Y)\big)$ equipped with super-fluxes $(G_4^s,\, G_7^s)$ of the form \eqref{SuperFluxDensitiesInIntroduction} which constitute a map of $L_\infty$-algebroids valued in the rational $4$-sphere
$$
(G_4^s, \, G_7^s) \, : \, TY^{11}\longrightarrow \mathfrak{l}S^4  .
$$
Each of these may be partially gauge-fixed as per \cite[Cor. 3.16]{Gi26}, bringing the corresponding $(11 \vert \mathbf{32})$-dimensional super-coframe in the form
$e_Y = (e, \, \Phi \, e')$ and $\psi_Y = \psi\,  + \, \lambda \, e'$ from \eqref{11DPartialGaugeFix}\eqref{11DPartialGaugeFixFermion} for some $U(1)$-connection $e'\in \mathrm{Conn}_{U(1)}(Y^{11})$, a $(10\vert \mathbf{16}\oplus \overline{\mathbf{16}})$-dimensional supercoframe $(e,\psi)$ \eqref{BasicSuperCoframe} and a dilaton/dilatino pair $(\Phi,\lambda)$ \eqref{DilatonDilatinoPair} on the base $X^{10}$ .

Decomposing the 11D super-spacetime torsion constraint in this gauge as per Lem. \ref{11DTorsionContraintIn10D} yields the expected 10D super-spacetime torsion constraint \eqref{10DTorsionConstraint}, the form of $F_2^s$ expansion from \eqref{IIASuperFluxDensities} and the expansion of the dilaton differential \eqref{DilatonDilatinoEquation}. In turn, expanding the $(G_4^s, G_7^s)$-fluxes \eqref{SuperFluxDensitiesInIntroduction} in this gauge and applying the reduction via cyclification of Prop. \ref{NonTrivialCircleBundleCyclification/Oxidation}, with respect to the associated connection $e'$, yields precisely (and bijectively, per gauge equivalence class) the form of expansions for the remaining fluxes $(H_3^s, F_4^s, F_6^s, H_7^s)$ from \eqref{IIASuperFluxDensities}, hence automatically forming a rational cyclified $4$-sphere cocycle
$$
\big(F_2^s, \, H_3^s, \, F_4^s, \,  F_6^s, \, H_7^s\big)\, : \, TX^{10}\longrightarrow \mathrm{cyc}(\mathfrak{l}S^4) \, .
$$
Although the  expansion of the given $(G_4^s, G_7^s)$-fluxes \eqref{SuperFluxDensitiesInIntroduction} in this gauge is algorithmic in nature, we include that of $G_7^s$ for exposition: 
\begin{align*}
G^s_7
    \, &=\,  
    \tfrac{1}{7!}
    (G_7)_{a_1 \cdots a_7}
    e^{a_1}_Y \cdots e^{a_7}_Y
    \;\,+\;\,
    \tfrac{1}{5!}
    \big(\,
    \overline{\psi}_Y
    \Gamma_{a_1 \cdots a_5}
    \psi_Y
    \big)
    e^{a_1}_Y \cdots e^{a_5}_Y \\
    &=\,  \sum_{0\leq a_{i}\leq 9}\tfrac{1}{7!}
    (G_7)_{a_1 \cdots a_7}
    e^{a_1} \cdots e^{a_7} \, + \, \sum_{0\leq a_{i}\leq 9}\tfrac{\Phi}{6!}
    (G_7)_{a_1 \cdots a_6 \ten}
    e^{a_1} \cdots e^{a_6} \, e' 
    \\ 
    & 
    + \sum_{0\leq a_{i}\leq 9}
    \!\!\tfrac{1}{5!}
    \Big(\,
    \overline{(\psi +\lambda e')}\,
    \Gamma_{a_1 \cdots a_5}\, 
    \big(\psi + \lambda e'\big)
    \!\Big)
    e^{a_1} \cdots e^{a_5} 
    +  \sum_{0\leq a_{i}\leq 9} \!\!
    \tfrac{\Phi}{4!}
    \Big(\,
    \overline{(\psi +\lambda e')}\, 
    \Gamma_{a_1 \cdots a_4 \ten }\,
    \big(\psi + \lambda e'\big) \!
    \Big)
    e^{a_1} \cdots e^{a_4}\, e' 
    \\
    & =\,  \sum_{0\leq a_{i}\leq 9} \Big( \tfrac{1}{7!} 
    (G_7)_{a_1 \cdots a_7}
    e^{a_1} \cdots e^{a_7} \;\,+\;\,
    \tfrac{1}{5!}
    \big(\,
    \overline{\psi}
    \Gamma_{a_1 \cdots a_5}
    \psi
    \big)  
    e^{a_1} \cdots e^{a_5} \Big)
    \\
    & \qquad + e'\, \sum_{0\leq a_{i}\leq 9}  \Big( \tfrac{\Phi}{6!}
    (G_7)_{a_1 \cdots a_6 \ten}
    e^{a_1} \cdots e^{a_6} \, \, - \, \, \tfrac{2}{5!} \big(\,\overline{\psi} \Gamma_{a_1\cdots a_5} \lambda\big) e^{a_1}\cdots e^{a_5} \, \, + \, \, \tfrac{\Phi}{4!} \big(\,\overline{\psi} \Gamma_{\ten a_1\cdots a_4} \psi\big) e^{a_1}\cdots e^{a_4}  \Big)   .
\end{align*}
This yields the forms of $H_7^s$ and $F_6^s$ in \eqref{IIASuperFluxDensities}, respectively, together with their relation to the (bosonic) components of the 11D $G_7^s$ super-flux. Totally analogously, one sees that in this gauge 
\begin{align*}
 G_4^s \, &= \,    
    \sum_{0\leq a_{i}\leq 9} \Big( \tfrac{1}{4!}(G_4)_{a_1 \cdots a_4} \, e^{a_1}\cdots e^{a_4} \;\,  
    +\;\, \tfrac{1}{2}\big(\,\overline{\psi} \Gamma_{a_1 a_2} \psi\big)  e^{a_1} e^{a_2} \Big) 
 \\ & \qquad + \,  e' \, \sum_{0\leq a_{i}\leq 9}\Big(  \tfrac{\Phi}{3!}(G_{4})_{a_1 a_2 a_3 \ten} \, e^{a_1}e^{a_2}e^{a_3} \,\, - \, \, 
    \big(\,
    \overline{\psi}
    \Gamma_{a_1 a_2}
    \lambda \big)
    e^{a_1}  e^{a_2} \, \,  \;\, +  \;\, \Phi \big(\,\overline{\psi} \Gamma_{\ten a} \psi\big) e^a \Big) 
\end{align*}
which yields the forms of the remaining super-fluxes $F_4^s$ and $H_3^s$. 

Finally, the total collection of super-fields $(e,\psi,\omega)$, $(\Phi,\lambda)$, and $(F_2^s, H_3^s, F_4^s, F_6^s, H_7^s)$ on $X^{10}$ necessarily restricts to the corresponding original 10D IIA supergravity solution on the underlying bosonic spacetime $\bosonic{X}^{10}\hookrightarrow X^{10}$ we started with, by construction, since the diagram of the underlying supermanifolds commutes
  \vspace{-1mm} 	
  \[
	\xymatrix@C=1.6em@R=.6em{ \bosonic{Y}^{11} \ar[d] \ar@{^{(}->}[rr] &&   Y^{11} \ar[d] 
		\\ 
		\bosonic{X}^{10} \ar@{^{(}->}[rr] && X^{10}
	\mathrlap{,} }
	\]
   \vspace{-2mm} 
\noindent hence exhibiting the cyc$(\mathfrak{l}S^4)$ super-fluxes as its unique (rheonomic) extension.
\end{proof}

\smallskip 
\noindent Summarizing the content of the theorem:
\smallskip 
\begin{equation}
  \label{10DTabledStatement}
  \hspace{-3mm} 
 \colorbox{lightgray}
  {\!\!\!\!\!\!\!
     \def\arraystretch{1.2}
    \begin{tabular}{c}
      $(10\vert\mathbf{16}\oplus \overline{\mathbf{16}})$-dimensional 
      super-spacetimes
      $\big(
        X, (e, \psi, \omega)
      \big)$
      \\
      carrying super-flux 
      $(F_2^s, H_3^s, F_4^s, F_6^s, H_7^s)$
      (from \eqref{IIASuperFluxDensities}),
      \\
      satisfying the $\mathrm{cyc}(\mathfrak{l}S^4)$ Bianchi identity
      (from \eqref{CycS4Bianchis}).
    \end{tabular}
  \!\!\!\!}
  \hspace{.1cm}
  \Leftrightarrow
  \hspace{.1cm}
  \colorbox{lightgray}{\!\!\!\!\!
    \begin{tabular}{c}
      Solutions of IIA 10D SuGra on $\bosonic{X}$
      \\
        with flux sources $(F_2, H_3, F_4)$
      \\
      and their dual fluxes.
    \end{tabular}
    \!\!\!\!}
\end{equation}

\medskip 
Now, let $\mathcal{A}$ be an admissible flux quantization law for 11D supergravity with $\mathfrak{l}\mathcal{A} \cong \mathfrak{l} S^4$, such as the actual topological 4-sphere $S^4$ (Hypothesis H) or the one described recently in \cite{BSS26}, applied on super-spacetime via \cite[Thm. 3.1]{GSS24-SuGra}.  
Then, since the topological cyclification descends to the rational cyclification under rationalization \cite[Thm. A]{VPB85} (amplified in our current context in \cite[Prop. 3.2]{FSS17}\cite{BMSS19}\cite{SV3}) 
$$
\mathfrak{l}\big(\mathrm{Cyc}(\mathcal{A})\big) \; \cong \;  \mathrm{cyc}(\mathfrak{l}\mathcal{A}) \, ,
$$
we obtain the following compatibility of flux quantization between the full 11D and 10D IIA supergravities. 

\begin{corollary}[\bf Compatibility with 11D Flux Quantization]\label{Compatibilitywith11DFluxQuantization}
Globally completed $\mathrm{Cyc}(\mathcal{A})$ flux-quantized 10D IIA SuGra field configurations (cf. Claim \ref{FluxQuantizedSuperFieldsOf10dSugra}) are in exact correspondence with (globally completed)
$S^1$-symmetric 11D $\mathcal{A}$ flux-quantized SuGra field configurations (cf. \cite[Claim 1.1]{GSS24-SuGra}) via 
\vspace{-1mm} 
\begin{equation}
\label{VerticalCyclificationOxidationOfFluxQuantizedFieldConfigurations}
\hspace{-2.2cm} 
  \begin{tikzcd}[
    row sep=32pt, 
    column sep=68pt,
    ampersand replacement=\& 
  ]
    \&[-20pt] \& \& \&[+5pt]
    {
      \mathcal{A}
    }
    \ar[
      d,
      "{ \mathbf{ch}^{\mathcal{A}} }"
    ]
    \ar[dd, <->, squiggly, purple, bend left=60, "{ \mathrm{Oxi/Cyc} }"{pos=0.3}]
    \\
    \&
    \;\;
    { Y }
    \;\;\;
    \ar[
      rr,
      "{ (G_4^s,G_7^s) }",
      "{
        \scalebox{.65}{
          \begin{tabular}{c}
            \color{darkgreen} \bf $S^1$-symmetric super-C-field flux \eqref{CurvedG4G7}
          \end{tabular}
        }
      }"{swap}
    ]
    \ar[
      urrr,
      bend left=10,
      "{
        \scalebox{.65}{
          \begin{tabular}{c}
            \color{darkgreen} \bf local C-field charge
          \end{tabular}
        }
      }"{sloped},
      "{ \rchi^{M} }"{sloped, pos=.48,swap}
    ]
    \ar[dd, <->, squiggly, purple, "{ \mathrm{Ext/Red} }"]
    \& \&
    {
      \Omega^1_{\mathrm{dR}}( -;\mathfrak{l}S^4 )_{\mathrm{clsd}}
    }
    \ar[
      r,
      "{ \eta^{\,\scalebox{.5}{$\shape$}} }"
    ]
    \&
    {
      \shape \, \Omega^1_{\mathrm{dR}}( -; \mathfrak{l}S^4 )_{\mathrm{clsd}}
    }
    \ar[dd, <->, squiggly, purple, bend right=40, "{ \mathrm{oxi/cyc} }"{swap, pos=0.3}]
    \\[20pt] 
    \&[-20pt]
    \& \&
    \&[+5pt]
    {
      \mathrm{Cyc}(\mathcal{A})
    }
    \ar[
      d,
      "{ \mathbf{ch}^{\mathrm{Cyc}(\mathcal{A})} 
      }"{xshift=0pt}
    ]
    \\
    \&
    \;\;
    { X }
    \;\;\;
    \ar[
      rr,
      "{ (F_2^s, H_3^s F_4^s,F_6^s, H_7^s) }",
      "{
        \scalebox{.65}{
          \begin{tabular}{c}
            \color{darkgreen} \bf super-NS/RR-field flux \eqref{CurvedNS/RRCocycles}
          \end{tabular}
        }
      }"{swap}
    ]
    \ar[
      urrr,
      bend left=10,
      "{
        \scalebox{.65}{
          \begin{tabular}{c}
            \color{darkgreen} \bf local NS/RR-field charge
          \end{tabular}
        }
      }"{sloped},
      "{ \rchi^{IIA} }"{sloped, pos=.48,swap}
    ]
    \& \&
    {
      \Omega^1_{\mathrm{dR}}( -; \mathrm{cyc}(\mathfrak{l}S^4) )_{\mathrm{clsd}}
    }
    \ar[
      r,
      "{ \eta^{\,\scalebox{.5}{$\shape$}} }"
    ]
    \&
    {
      \shape \, \Omega^1_{\mathrm{dR}}( -; \mathrm{cyc}(\mathfrak{l}S^4) )_{\mathrm{clsd}} \, ,
    }
  \end{tikzcd}
\end{equation}
where, for brevity, we have omitted the (higher) gauge fields -- homotopies $(\widehat{C}_3^s, \widehat{C}_6^s)$ and $(\widehat{A}_1^s,\widehat{B}_2^s,\widehat{A}_3^s,\widehat{A}_5^s,\widehat{B}_6^s)$ filling each of the diagrams, respectively.  
\end{corollary}

\begin{remark}[\bf A curious automorphism of the cyclified $L_\infty$-algebra]
It turns out that the cyclified 4-sphere IIA classifying $L_\infty$-algebra enjoys a non-trivial, non-linear automorphism 
\begin{align*}
Z_{\mathrm{IIA}} \, : \, \mathrm{cyc}(\mathfrak{l}S^4) \xrightarrow{\qquad \sim \qquad}\, &  \mathrm{cyc}(\mathfrak{l}S^4) 
\\[-2pt]
\omega_2 \qquad \longmapsfrom \qquad  &  \,  \omega_2 
\\[-1pt]
sg_4 \qquad \longmapsfrom \qquad  &  \,   sg_4 
\\
 g_4 + (\omega_2)^2 \qquad \longmapsfrom \qquad  &  \,  g_4 
\\
sg_7 - \omega_2 \, g_4 -\tfrac{1}{2} (\omega_2)^3\qquad \longmapsfrom \qquad  &  \,  sg_7 
\\
g_7 + g_4 \, sg_4 g_4 + (\omega_2)^2 \, sg_4  \qquad \longmapsfrom \qquad  &  \,   g_7  \, ,
\end{align*}
\noindent
as can be easily checked by verifying that the corresponding differential relations (``Bianchi identities'') are preserved. 
It follows, then, that post-composing the IIA NS/RR-fluxes from Eq. \ref{IIASuperFluxDensities} with this automorphism  yields new cyclified 4-sphere cocycles on the same super-spacetime $\big(X^{10},(e,\psi,\omega)\big)$:
$$
Z_{\mathrm{IIA}} \, \circ \, (F_2^s, H_3^s, F_4^s, F_6^s, H_7^s) \;\; : \;\; 
TX^{10} \xrightarrow{\quad\quad }\mathrm{cyc}(\mathfrak{l}S^4) \xrightarrow{\quad} \mathrm{cyc}(\mathfrak{l}S^4)\,.
$$
We stress, however, that the resulting $\mathrm{cyc}(\mathfrak{l}S^4)$-cocycles will have a different component-form when explicitly expanded in the super-coframe basis, and hence will \text{not} correspond to new solutions of the same IIA SuGra equations of motion. That is, this is \textit{not} to be interpreted as a newly discovered \textit{bona-fide symmetry} of IIA SuGra. Instead, solving the superspace Bianchi identities for these new form of NS/RR-super-fluxes will result into a different set of field equations, but which will necessarily have an isomorphic space of solutions to  IIA supergravity (by composing with the inverse automorphism)! Hence, at best, the resulting field equations might be called a new (on-shell) ``dual'' formulation of 10D IIA supergravity.  We leave the study of the explicit form of the field equations, and  physical importance (or the contrary), of this new formulation for future works.
\end{remark}






\vspace{-6mm}
\paragraph{Comparison to other 10D IIA superspace formulations.} 
Apart from the usual range in conventions and field redefinitions of IIA supergravity\footnote{Such as extracting prefactors proportional to $\Phi$ in front of the induced 10D coframe $e$, and similarly in the corresponding splitting of the 11D gravitino $\psi_Y$. The allowed possibilities in the ``Einstein frame'' are summarized in 
\cite[(B.9)]{BNS04}, following \cite{LPSS}. In particular, scaling  $e \mapsto \Phi^{-\tfrac{1}{8}} e$ etc., should be in line with the original sources (\cite{CW84}\cite{HN85}) of the 11D SuGra $\rightarrow $ 10D IIA SuGra reduction on the underlying bosonic spacetime $\bosonic{X}\hookrightarrow X$.}, and further expansions of the formulas above\footnote{For instance, one explicitly splits $\psi$ and $\lambda$ into the corresponding Spin$(1,9)$-irreps. 
Since  $\Gamma^{\ten}$ acts as a projector on each irrep, 
 we may further expand the pairings appearing via $(\overline{\psi} \Gamma^{\ten} \psi) = - (\overline{\psi}_L \psi_R) + (\overline{\psi}_R \psi_L)  $  etc. 
 (cf. \cite[Lem. 3.1]{GSS24-TDuality}).}, we highlight the following essential differences from the IIA superspace formulations already reported in the literature, at least at the level of local coordinate charts:

\begin{itemize}[leftmargin=.8cm]
 \item[\bf (i)] The formulation of \cite{DFGT} 
seems closer to ours in spirit, as it 
may be interpreted as employing a different classifying $L_\infty$-algebra\footnote{The authors in the above source use the terminology ``free differential algebra'' instead.} defined with different generators. In particular, there appears to be no (dual) RR super 6-form flux $F_6^s$, nor a (dual) NS super 7-form flux $H_7^s$ \cite[(2.8)-(2.9)]{DFGT} -- hence it is not manifestly duality-symmetric. Curiously, the claim therein (Appendix B) is that their formulas follow by dimensional reduction of a duality-symmetric 11D cocycle similar to ours from \eqref{CurvedG4G7}. However, it seems to us that their extra imposition of a \textit{choice} of ``rheonomic parametrization'' is implemented \textit{prior to dimensional reduction}, hence eliminating some of the fluxes. In our approach, the rheonomic parametrization is already implicit in the $S^4$-cocycle condition in 11D, and properly taking its full dimensional reduction via cyclification results instead into a $\mathrm{cyc}(S^4)$-cocycle, as expected by the flat tangent-space wise case. A similar ``string frame'' IIA superspace formulation is reported in \cite{NO11}, where it is further claimed (cf. p. 5, (A.1) therein) that the original source by \cite{DFGT} had missed several terms in their corresponding rheonomic parametrization of the curvature forms.

 \item[\bf (ii)] The paper \cite[Appendix C]{HS05} explicitly presents calculations of dimensional reduction on 11D superspace, albeit only reducing the corresponding $G_4^s$. The formulas they obtain for $F_2^s, H_3^s$ and $F_4^s$ are compatible with ours from \eqref{IIASuperFluxDensities}. These seem further compatible with the super-coframe expansions of the super-fields reported earlier for IIA superspace supergravity in \cite{CGO87} where, however, the method of derivation is not explained.

 \item[\bf (iii)] Finally, the authors in \cite{CGNSW97} follow the above sources and introduce \textit{by hand all} the dual NS and RR fluxes, hence bringing their formulation to a duality-symmetric form. However, their form of the dual NS flux $H_7^s$ \cite[(3.11)]{CGNSW97} explicitly differs from the (natural) dimensionally reduced form of ours from \eqref{IIASuperFluxDensities}. Indeed, they report a non-trivial component along the ``$\psi \wedge e^6$''-leg which, as shown in the proof of Thm. \ref{10dIIASugraEoMFromSuperFluxBianchiIdentity}, simply cannot appear by the dimensional reduction $G_7^s$.\footnote{We make no claim on the validity of the overall formulation therein, but simply that it is not equivalent to ours.}
\end{itemize}


\medskip
\noindent
Let us now close this comparative discussion by stressing the advantages and structural improvements of our approach to the super-spacetime formulation of 10D IIA, and further supergravities in general:

\smallskip 
\begin{itemize}[leftmargin=.8cm]
\item[\bf (i)] Firstly, our superspace formulations for 11D and its descendant IIA have been explicitly proven to be \textit{bi-directionally equivalent} to the corresponding theories on the underlying bosonic spacetimes. This is in contrast to what is recorded in existing literature, which is only the implication of the spacetime equations. All further superspace dimensional reductions in our formalism (see \S \ref{OutlookSection}) will manifestly enjoy the same property.

\item[\bf (ii)] Secondly, compared to all existing approaches, and as described in much more detail in  \S\ref{IntroductionSection}, our viewpoint starts from an a priori (conceptual) splitting of the fields involved, namely: 
\begin{itemize}
    \item[{\bf(a)}] Those describing a bare super-spacetime $(X, e, \psi, \omega)$ via the super-coframe and the Spin connection, namely as a (super-)Cartan geometry for the (super-)subgroup inclusion $\mathrm{Spin}(1,d)\hookrightarrow \mathrm{ISO}(\FR^{1,d\vert \mathbf{N}})$, with vanishing torsion in the bosonic direction, potentially equipped with the corresponding dilaton/dilatino pair $(\Phi,\lambda)$ due to winding along the higher-dimensional theory's fiber.
\item[{\bf(b)}] The higher form-fluxes entering (initially) only through their \textit{curvatures}, whose form can be deduced by the corresponding super-invariant cocycles which exist on the flat spacetime model $L_\infty$-algebra. The latter are all determined by that of 11D, the rational 4-sphere $\mathfrak{l}S^4$, its cyclification for 10D, and further its toroidifications for lower dimensions (see \S\ref{OutlookSection}).   
\end{itemize}
\end{itemize} 

\medskip

\noindent {\bf Fate of the remaining fluxes under non-abelian flux quantization.}
We include here a brief discussion regarding the above-advocated flux quantization. 

   \begin{itemize}
        \item[\bf (i)]
Contrary to the standard twisted $K$-theory suggestion, dimensional reduction of the 11D topological flux content necessarily 
makes the $H_7$ field appear, hence contributing non-trivially to the admissible full topological flux quantization laws of the IIA supergravity field content (cf.  \cite{BMSS19}). 

 \item[\bf (ii)]
 Moreover, the $F_8$ field does not directly appear upon dimensional reduction from 11D; hence, its existence (being in fact the Hodge dual to $F_2$) does not impose conditions on the allowable flux quantization laws for the field content. It should only be defined -- a posteriori to flux quantization -- as the Hodge dual of $F_2$.

 \item[\bf (iii)]
From our super-spacetime perspective, rather than this being a surprise, it is a confirmation of the special nature of $F_2$ and $F_8$, namely that they are of ``gravitational origin'', the former arising via the $S^1$-fiber aligned coframe component of the 11D coframe, which is ``shrunk'' yielding the 10D super-spacetime, and the latter being the Hodge dual of $F_2$. Our formulation also automatically detects $F_2$ being already flux-quantized in $\mathrm{BU}(1)$, as expected by the underlying $S^1$-principal bundle topology, and hence further its Hodge dual should not receive -- a posteriori -- any additional independent flux quantization laws.
\end{itemize} 

\section{Outlook}\label{OutlookSection}

We have established, via the natural superspace formulation of Thm. \ref{10dIIASugraEoMFromSuperFluxBianchiIdentity}, that the cyclification $\mathrm{cyc}(\mathfrak{l}S^4)$ turns out to be the rational classifying $L_\infty$-algebra of 10D IIA supergravity. This suggests that for \textit{all descendant supergravity theories} arising as reductions of the parent maximal 11D supergravity theory, the corresponding classifying $L_\infty$-algebras should also be systematically determined by the \textit{toroidifications} of the parent classifying rational 4-sphere $\mathfrak{l}S^4$ (see \cite{SV1}\cite{SV2}\cite{GSS24-TDuality})
$$
\mathrm{tor}^k(\mathfrak{l}S^4) \, , 
$$
to be justified by a corresponding tree of reduced superspace formulations obtained by an appropriately modified ``toroidification'' version of Prop. \ref{NonTrivialCircleBundleCyclification/Oxidation}. 

\smallskip 
Then, since for any space $\mathcal{A}$ the rationalization of its $k$-toroidifications, $\mathrm{Tor}^k \mathcal{A} := \mathrm{Map}\big(T^k, \mathcal{A} \big) /\!/ T^k$, coincides with the (rational) $k$-toroidification of its rationalization \cite[Thm 2.6]{SV-loop}	
$$
\mathfrak{l}\, \mathrm{Tor}^k(\mathcal{A}) \, \cong \, \mathrm{tor}^{k}(\mathfrak{l} \mathcal{A}) \, ,
$$
this naturally leads to the following generalized hypothesis for the manifestly compatible tree of flux quantizations, in appropriately \textit{$k$-toroidified} non-abelian cohomology theories, for all descendant theories of 11D supergravity: 

\medskip 
\hspace{-8mm}
\adjustbox{
  margin=1pt,
  bgcolor=lightgray
}{\!\!\!\!\!
\def\arraystretch{1.1}
\begin{tabular}{l} 
\textit{Given a single choice of an admissible non-abelian flux quantization law $\mathcal{A}$ for the C-field in 11D super- }
\\
\textit{gravity, hence of rational homotopy type that of the 4-sphere, then all flux quantizations of the resulting }
\\
\textit{dimensionally-reduced higher gauge fields in each descendant supergravity theory in (11-k)D, obtained by }
\\
\textit{reduction on a $k$-torus principal bundle, are given by corresponding full toroidifications }
\\
\hspace{8cm} $
\mathrm{Tor}^k(\mathcal{A})\, .
$
\\
\textit{It follows, then, that all of these globally completed (higher gauge) supergravities are automatically com- }
\\
\textit{patible upon re-oxidation and toroidification. }
\end{tabular}
\!\!\!\!\!\!}



\medskip 
Indeed, the justification of this proposal follows by a straightforward (but computationally cumbersome) generalization of our currently established results. We briefly outline how this works in the following, without recording the explicit calculations and proofs here.

\vspace{-1mm} 
\paragraph{Toroidal reduction of flux fields via toroidification of classifying $L_\infty$-algebras.}
Firstly, it should not be surprising that an immediate generalization of our curved oxidation/cyclification result from Prop. \ref{NonTrivialCircleBundleCyclification/Oxidation} to the case of toroidal super-principal $T^k=U(1)^{\times k}$-bundles $Y\rightarrow X$ is within reach. In particular, for any choice of $U(1)^{\times k}$-connection $e'=\big(\, {e'}_1,\cdots,{e'}_k\,\big) \in \Omega^{1}_\mathrm{dR}(Y;\FR^k)$ of curvature $F_2= \dd e'$, this will yield a bijection between (cf. the tangent-space case from \cite[Prop. 2.60]{GSS24-TDuality}) 

\vspace{1mm} 
\begin{itemize} 
\item[\bf (a)] $T^k$-invariant maps of super-$L_\infty$-algebroids out of the tangent Lie algebroid $TY$ into $\mathfrak{h}$,  
  \vspace{1mm} 
\item[\bf (b)] maps of super-$L_\infty$-algebroids out of $TX\cong T\big(Y/T^k\big)$ into $\mathrm{tor}^k(\mathfrak{h})$ that preserve the 
2-form $F_2$,
  \begin{equation}
    \label{NontrivialLinftyAlgebroidToroidificationHomIsomorphism}
    \hspace{-1cm} 
    \Big\{\!\!
    \begin{tikzcd}
      TY
      \ar[
        rr,
        "{ F }"
      ]
      &&
      \mathfrak{h}
    \end{tikzcd}
   \! \!\Big\}
    \begin{tikzcd}[
      column sep=85pt
    ]
      \ar[
        r,
        shift left=5pt,
        "{  \scalebox{.7}{\color{darkgreen}
            \bf
            reduction}\;\;
          \mathrm{rdc}_{F_2}
        }",
        "{ \sim }"{swap, yshift=-2pt}
      ]
      \ar[
        r,
        <-,
        shift right=5pt,
        "{ \scalebox{.7}{
            \color{darkgreen}
            \bf
            oxidation
          }\;\;
          \mathrm{oxd}_{F_2}
        }"{swap},
      ]
      &
      {}
    \end{tikzcd}
    \bigg\{\!\!\!
    \begin{tikzcd}[row sep=-3pt, column sep=large]
      TX
      \ar[
        rr,
        "{ \widetilde F }"
      ]
      \ar[
        dr,
        "{ F_2 }"{swap}
      ]
      &&
      \mathrm{tor}^k(\mathfrak{h})
      \ar[
        dl,
        "{ \omega_2 }"
      ]
      \\
      &
      b \mathbb{R}^k
    \end{tikzcd}
    \!\!\!\bigg\}.
  \end{equation}

That is, a bijection\footnote{Which, analogously to part {\bf{(ii)}} of Prop. \ref{NonTrivialCircleBundleCyclification/Oxidation}, will extend to an isomorphism of the corresponding integrating $\infty$-groupoids of higher gauge potentials and transformations.} of $T^k$-invariant, closed $\mathfrak{h}$-valued forms on $Y$ and closed $\mathrm{cyc}(\mathfrak{h})$-valued forms on $X\cong Y/T^k$ with fixed curvature 2-form component $F_2 := \dd e'$
$$
\Omega^1_{\mathrm{dR}}\big(Y; \,\mathfrak{h}\big)_{\mathrm{clsd},\, \mathrm{inv.}} \; \cong_{e'} \; \Omega^1_{\mathrm{dR}}\big(X; \,\mathrm{tor}^k(\mathfrak{h})\big)_{ \mathrm{clsd}, \, F_2}\, .
$$
\end{itemize}

\vspace{-3mm} 
\paragraph{SuGra in $(11-k)D$ via $\mathrm{tor}^k(\mathfrak{l} S^4)$-valued cocycles on super-spacetime.}
Then, starting from any $T^k$-$(11\vert \bf{32})$ super-spacetime $(Y, (e_Y,\psi_Y,\omega_Y)\big)$ with a $T^k$-principal bundle structure, due to the abelian structure of the group, we may 
again employ the result of \cite[Cor. 3.16]{Gi26}. That is, we may again partially gauge fix along the lines of pp. 12-13, reducing the structure group to 
$$
\mathrm{Spin}(1,10-k)\longhookrightarrow\mathrm{Spin}(1,10)\, .
$$
This process will, in particular, produce the $T^k$-connection (see \cite{Gi26} for details)
$$
e' \, \in \, \Omega^{1}_\mathrm{dR}(Y; \FR^k)
$$ 
necessary to pass invariant and closed $\mathfrak{l}S^4$-valued flux-forms through the reduction-by-toroidification of the bijection of Eq. \eqref{NontrivialLinftyAlgebroidToroidificationHomIsomorphism}.

\smallskip 
Finally, then, we may employ this to reproduce a version of Thm. \ref{10dIIASugraEoMFromSuperFluxBianchiIdentity}. Starting from any $T^k$-symmetric on-shell 11D supergravity configuration, (rheonomically) expressed as a $\mathfrak{l}S^4$-cocycle $(G_4^s, G_7^s)$ of the form \eqref{CurvedG4G7}, we pass through the isomorphism  \eqref{NontrivialLinftyAlgebroidToroidificationHomIsomorphism} producing a $\mathrm{tor}^k(\mathfrak{l}S^4)$-cocycle on the base super-spacetime $X$, of a particular form of super-coframe component expansion (cf. \eqref{IIASuperFluxDensities}). We do not attempt to write these down explicitly in this outlook, but we stress that they will follow the same principles as noted for the 11D (see \eqref{CurvedG4G7}) and IIA (see \eqref{CurvedNS/RRCocycles}) curved and global super-cocycles. Namely, they will form a (uniquely determined) combination of 

\begin{itemize} 
\item[\bf{(i)}] a collection of flux densities on {\color{darkblue} bosonic but non-trivial} underlying manifolds $\bosonic{X}\hookrightarrow X$, which satisfy the corresponding duality-symmetric tor$^k(\mathfrak{l}S^4)$-Bianchi identities, 

\item[\bf(ii)] which globally extend the canonical supersymmetric contributions arising on the corresponding {\color{darkorange} flat but super-}tangent spaces\footnote{Here the superscript on ${\bf{32}}^{\scriptscriptstyle k\hookrightarrow 11}$ denotes each of the corresponding $\mathrm{Spin}(1,10-k)$-representations, into which the original Majorana $\bf{32}$ in 11D branches to for each $k$.}
$$
\FR^{(11-k)\vert {\bf{32}}^{\scriptscriptstyle k\hookrightarrow 11}}, 
$$ 
satisfying the analogous identities; 

\item[\bf{(iii)}] 
along with certain additional subtle ``{\color{olive} mixed component extensions}'', reflecting the fact that the total global extensions are really reductions of 11D fluxes, whose super-coframe expansions {\color{olive}wind non-trivially} along the toroidal $T^k$-fibers.
\end{itemize} 

In any case, the main point here is that the space of such closed $\mathrm{tor}^k(\mathfrak{l}S^4)$-valued forms on $X$ will be automatically  \textit{equivalent} (bi-directionally!) to the space of (full) on-shell backgrounds of the corresponding $(11-k)D$ supergravity theory on $\bosonic{X}\hookrightarrow X$, thus suggesting and allowing for the choice of flux quantization in the topological toroidification $\mathrm{Tor}^k(\mathcal{A})$ of the 11D non-abelian cohomology theory $\mathcal{A}$.   

\vspace{.5cm} 
\noindent {\bf Acknowledgements.} \,  
We thank Urs Schreiber for useful discussions.

\medskip


\end{document}